\documentclass{lmcs}
\usepackage[utf8]{inputenc}
\pdfoutput=1

\usepackage{lastpage}
\lmcsdoi{18}{3}{23}
\lmcsheading{}{\pageref{LastPage}}{}{}%
{Dec.~01,~2020}{Aug.~17,~2022}{}

\usepackage{etoolbox}
\apptocmd{\sloppy}{\hbadness 10000\relax}{}{}


\usepackage{amssymb}
\usepackage{mathpartir}
\usepackage{url}
\usepackage{mathtools}
\usepackage{bussproofs}
\usepackage{pdftexcmds}
\usepackage{stmaryrd}
\usepackage{xspace}

\newcommand{\subst}[2]{\left[{#2}\middle/{#1}\right]}

\newcommand{\vect}[1]{\overrightarrow{#1}}

\newcommand{\mathcd}[1]{\mathtt{\color{blue}#1}}

\newcommand{\bbN}[0]{\mathbb{N}}

\newcommand{\cC}[0]{\mathcal{C}}
\newcommand{\cD}[0]{\mathcal{D}}

\newcommand{\cN}[0]{\mathcal{N}}
\newcommand{\cP}[0]{\mathcal{P}}

\newcommand{\cR}[0]{\mathcal{R}}
\newcommand{\cS}[0]{\mathcal{S}}
\newcommand{\cT}[0]{\mathcal{T}}

\newcommand{\CR}[0]{\cC\cR}
\newcommand{\CRi}[0]{\mathrm{CR1}}
\newcommand{\CRii}[0]{\mathrm{CR2}}
\newcommand{\CRiii}[0]{\mathrm{CR3}}

\newcommand{\SN}[0]{\cS\cN}
\newcommand{\NT}[0]{\cN\cT}

\newcommand{\Red}[1]{\mathsf{Red}_{#1}}
\newcommand{\topset}[1]{#1^{\top}}
\newcommand{\toptopset}[1]{#1^{\top\top}}

\newcommand{\imp}[0]{\to}
\newcommand{\red}[0]{\mathrel{\leadsto}}
\newcommand{\ared}[1]{\mathrel{\stackrel{#1}{\leadsto}}}

\newcommand{\ured}[0]{\mathrel{\breve{\leadsto}}}

\newcommand{\sem}[1]{\left\llbracket {#1} \right\rrbracket}

\def\then{\mathrel{\Longrightarrow}}

\def\pure{\overline}

\def\orelse{\mathbin{~|~}}
\def\emptyoset{\emptyset}
\def\emptymset{\mho}
\newcommand{\setlit}[1]{\mathopen{\textbf{\upshape\{}}{#1}\mathclose{\textbf{\upshape\}}}}
\newcommand{\msetlit}[1]{\Lbag{#1}\Rbag}

\def\ocup{\mathbin{\cup}}
\newcommand{\comprehension}[1]{\bigcup\setlit{#1}}
\newcommand{\mcomprehension}[1]{\biguplus\msetlit{#1}}
\newcommand{\isempty}{\mathcd{empty}}
\newcommand{\tuple}[1]{\langle{#1}\rangle}
\def\distinct{\delta}
\def\promote{\iota}

\def\kwsize{\mathrm{size}}

\def\kwdo{\mathcd{do}}
\def\kwlet{\mathbf{let}}
\def\kwlin{\mathbf{in}}
\def\kwcases{\mathbf{cases}}
\def\kwof{\mathbf{of}}
\def\kwinl{\mathbf{inl}}
\def\kwinr{\mathbf{inr}}
\def\plwhere{\mathcd{where}}

\def\bagwhere{\mathcd{where}_{\mathsf{bag}}}
\def\plempty{\mathcd{empty}}

\def\bagempty{\mathcd{empty}_{\mathsf{bag}}}
\def\id{\mathsf{id}}

\newcommand{\metaset}[1]{\{ {#1} \}}
\newcommand{\ibox}[1]{\mathop{%
   \left[ {#1} \right]
   }}
\newcommand{\compop}[1]{\mathop{%
   \ooalign{%
      \hfil \raise .2ex\hbox {$\scriptscriptstyle#1$}\hfil \crcr
      {$\bigcirc$}%
   }}}
\newcommand{\sz}[1]{\left\Vert #1 \right\Vert}
\newcommand{\len}[1]{\left\vert #1 \right\vert}
\newcommand{\erase}[1]{\left\lfloor #1 \right\rfloor}

\def\kwsel{{\color{blue}\mathtt{SELECT}}}
\def\kwdist{{\color{blue}\mathtt{DISTINCT}}}
\def\kwfrom{{\color{blue}\mathtt{FROM}}}

\def\kwunion{{\color{blue}\mathtt{UNION}}}

\def\kwall{{\color{blue}\mathtt{ALL}}}

\def\kwtrue{{\mathsf{true}}}
\def\kwfalse{{\mathsf{false}}}

\def\kwby{{\color{blue}\mathtt{BY}}}
\def\kworder{{\color{blue}\mathtt{ORDER}}}

\def\boolty{\mathbf{B}}

\newcommand{\maxred}[1]{\nu(#1)}

\DeclareMathOperator{\BV}{BV}
\DeclareMathOperator{\FV}{FV}

\DeclareMathOperator{\dom}{dom}

\DeclareMathOperator{\supp}{supp}

\def\theNRC{\mathit{NRC}}
\def\NRC{$\theNRC$\xspace}
\def\NRCl{$\theNRC_\lambda$\xspace}
\def\NRCsb{$\theNRC(\mathit{Set},\mathit{Bag})$\xspace}
\def\NRClsb{$\theNRC_\lambda(\mathit{Set},\mathit{Bag})$\xspace}

\def\NRCldi{$\theNRC_{\lambda\distinct\promote}$\xspace}

\definecolor{Gold}{RGB}{255,224,0}
\newcommand{\goldbox}[1]{\colorbox{Gold}{#1}}

\newenvironment{lemmaproof}[2]
{\begin{proof}[Proof of Lemma~\ref{#1}]
\emph{#2}

\medskip

}
{
\end{proof}
}

\begin{document}

\title{Strongly Normalizing Higher-Order Relational Queries}
%\titlerunning{Dummy short title} %TODO optional, please use if title is longer than one line

\author[W.~Ricciotti]{Wilmer Ricciotti\lmcsorcid{0000-0002-2361-8538}}
\author[J.~Cheney]{James Cheney\lmcsorcid{0000-0002-1307-9286}}
\address{
  Laboratory for Foundations of Computer Science,
  University of Edinburgh,
  United Kingdom
\and
The Alan Turing Institute,
London, United Kingdom
}
\email{research@wilmer-ricciotti.net, jcheney@inf.ed.ac.uk}

%%%%%%%%%%%%%%%%%%%%%%%%%%%%%%%%%%%%%%%%%%%%%%%%%%%%%%%%%%%%%%%%%%%%%%%%%%%

%% the abstract has to PRECEDE the command \maketitle:
%% be sure not to issue the \maketitle command twice!
\begin{abstract}
  Language-integrated query is a
  powerful programming
  construct allowing database queries and ordinary program code
  to interoperate seamlessly and safely.  Language-integrated query
  techniques rely on classical results about
  the nested
  relational calculus,
  stating that its queries can be algorithmically translated to SQL, as long as their result type is a flat relation.
  Cooper and others
  advocated \emph{higher-order} nested relational calculi as a basis
  for language-integrated queries in functional languages such as Links and F\#.
  However, the translation of higher-order relational queries to SQL relies on a rewrite system for which no \emph{strong normalization} proof has been published:
  a previous proof attempt does not deal correctly
  with rewrite rules that duplicate subterms.  This paper fills the gap in the
  literature, explaining the difficulty with a previous proof attempt,
  and showing how to extend the \emph{$\top\top$-lifting} approach of
  Lindley and Stark to accommodate duplicating rewrites.  We also
  show how to extend the proof to a recently-introduced calculus for
  \emph{heterogeneous} queries mixing set and multiset semantics.
\end{abstract}

\maketitle

\section{Introduction}
The Nested Relational Calculus~(\NRC)~\cite{BNTW95} provides a principled
foundation for integrating database queries into programming languages.  Wong's conservativity theorem~\cite{wong96jcss} generalized the classic
flat-flat theorem~\cite{ParedaensG92} to show that for any nesting depth $d$, a query
expression over flat input tables returning collections of depth at
most $d$ can be expressed without constructing intermediate results of
nesting depth greater than $d$.  In the special case $d=1$, this
implies the flat-flat theorem, namely that a nested relational query
mapping flat tables to flat tables can be expressed in a semantically equivalent way using
the flat relational calculus.
In addition, Wong's proof technique was constructive, and gave an
easily-implemented terminating rewriting algorithm for normalizing NRC queries to flat relational queries, which can in turn be easily translated to idiomatic SQL queries.
The basic approach has been extended
in a number of directions, including to allow for (nonrecursive)
higher-order functions in queries~\cite{Cooper09}, and to allow for
translating queries that return nested results to a bounded number of
flat relational queries~\cite{cheney14sigmod}.

Normalization-based techniques are used in language-integrated query systems such as
Kleisli~\cite{wong:comprehensions} and Links~\cite{CLWY06}, and can
improve both performance and reliability of language-integrated query
in F\#~\cite{cheney13icfp}.  However, most work on normalization
considers \emph{homogeneous} queries in which there is a single
collection type (e.g.\ homogeneous sets or multisets).
Currently, language-integrated query systems such as C\# and
F\#~\cite{meijer:sigmod} support duplicate elimination via a
$\mathtt{Distinct()}$ method, which is translated to SQL queries in an ad
hoc way, and comes with no guarantees regarding completeness or
expressiveness as far as we know, whereas Database-Supported Haskell
(DSH)~\cite{SIGMOD2015UlrichG} supports duplicate elimination but
gives all operations list semantics and relies on more sophisticated
SQL:1999 features to accomplish this.  Fegaras and Maier~\cite{DBLP:journals/tods/FegarasM00}
propose optimization rules for a nested object-relational calculus with set and bag constructs
but do not consider the problem of conservativity with respect to flat queries.

Recently, we considered a \emph{heterogeneous} calculus for mixed set
and bag queries~\cite{ricciotti19dbpl}, and conjectured that it too
satisfies strong normalization and conservativity theorems.  However,
in attempting to extend Cooper's proof of normalization we discovered
a subtle problem,  which makes the original proof incomplete.

Most techniques to prove the strong normalization property for higher-order
languages employ logical relations; among these, the Girard-Tait
\emph{reducibility} relation is particularly influential: reducibility
interprets types as certain sets of strongly normalizing terms enjoying
desirable closure properties with respect to reduction, called
\emph{candidates of reducibility}~\cite{GLT89}. The fundamental theorem
then proves that every well-typed term is reducible, hence also strongly
normalizing. In its traditional form, reducibility has a limitation that
makes it difficult to apply it to certain calculi: the elimination form of
every type is expected to be a \emph{neutral} term or, informally, an
expression that, when placed in an arbitrary evaluation context, does not
interact with it by creating new redexes. However, some calculi possess
\emph{commuting conversions}, i.e.\ reduction rules that apply to nested
elimination forms: such rules usually arise when the elimination form for a
type (say, pairs) is constructed by means of an auxiliary term of
any arbitrary, unrelated type. In this case, we expect nested elimination forms to
commute; for example, we could have the following commuting conversion hoisting
the elimination of pairs out of case analysis on disjoint unions:
\[
\begin{array}{l}
\kwcases~(\kwlet~(a,b) = p~\kwlin~t)~\kwof~\kwinl(x) \Rightarrow u; \kwinr(y) \Rightarrow v
\\
\red
\kwlet~(a,b) = p~\kwlin~\kwcases~t~\kwof~\kwinl(x) \Rightarrow u; \kwinr(y) \Rightarrow v
\end{array}
\]
where $p$ has type $A \times B$, $t$ has type $C + D$, $u,v$ have type $U$,
and the bound variables $a,b$ are chosen fresh for $u$ and $v$.
Since in the presence of commuting conversions elimination forms are not neutral, a straightforward adaptation of reducibility to such languages is precluded.

\subsection{\texorpdfstring{$\top\top$}{⊤⊤}-lifting and \texorpdfstring{\NRCl}{NRC-λ}}
Cooper's \NRCl~\cite{Cooper09a,Cooper09} extends the simply typed lambda calculus with collection types whose elimination form is expressed by \emph{comprehensions} $\comprehension{M | x \gets N}$, where $M$ and $N$ have a collection type, and the bound variable $x$ can appear in $M$:

\begin{prooftree}
\AxiomC{$\Gamma \vdash N : \setlit{S}$}
\AxiomC{$\Gamma, x : S \vdash M : \setlit{T}$}
\BinaryInfC{$\Gamma \vdash \comprehension{M | x \gets N} : \setlit{T}$}
\end{prooftree}
(we use bold-style braces $\setlit{\cdot}$ to indicate collections as expressions or types of \NRCl). In the rule above, we typecheck a comprehension destructuring collections of type $\setlit{S}$ to produce new collections in $\setlit{T}$, where $T$ is an unrelated type: semantically, this corresponds to the union of all the collections $M\subst{x}{V}$, such that $V$ is in $N$. According to the standard approach, we should attempt to define the reducibility predicate for the collection type $\setlit{S}$ as:
\[
\Red{\setlit{S}} \triangleq \metaset{N : \forall x,T,\forall M \in \Red{\setlit{T}}, \comprehension{M | x \gets N} \in \Red{\setlit{T}}}
\]
(we use roman-style braces $\metaset{\cdot}$ to express metalinguistic sets). Of course the definition above is circular, since it uses reducibility over collections to express reducibility over collections; however, this inconvenience could in principle be circumvented by means of impredicativity, replacing $\Red{\setlit{T}}$ with a suitable, universally quantified candidate of reducibility (an approach we used in~\cite{Ricciotti2017a} in the context of justification logic). Unfortunately, the arbitrary return type of comprehensions is not the only problem: they are also involved in commuting conversions, such as:
\[
\comprehension{M | x \gets \comprehension{N | y \gets P}}
\red
\comprehension{\comprehension{M | x \gets N} | y \gets P} \qquad (y
\notin FV(M))
\]
Because of this rule, comprehensions are not neutral terms, thus we cannot use the closure properties of candidates of reducibility (in particular, $\CRiii$~\cite{GLT89}) to prove that a collection term is reducible.
To address this problem, Lindley and Stark proposed a revised notion of reducibility based
on a technique they called $\top\top$-lifting~\cite{LindleyS05}. $\top\top$-lifting, which
derives from Pitts's related notion of $\top\top$-closure~\cite{Pitts98}, involves
quantification over arbitrarily nested, reducible elimination contexts (\emph{continuations});
the technique is actually composed of two steps: $\top$-lifting, used to define the set
$\topset{\Red{T}}$ of reducible continuations for collections of type $T$ in terms of $\Red{T}$, and
$\top\top$-lifting proper, defining $\Red{\setlit{T}} = \toptopset{\Red{T}}$ in terms of
$\topset{\Red{T}}$. In our setting, if we use $\SN$ to denote the set of strongly normalizing
terms, the two operations can be defined as follows:
\begin{align*}
\topset{\Red{T}} & \triangleq \metaset{ K : \forall M \in \Red{T}, K[\setlit{M}] \in \SN }
\\
\toptopset{\Red{T}} & \triangleq \metaset{ M : \forall K \in \topset{\Red{T}}, K[M] \in \SN}
\end{align*}
Notice that, in order to avoid a circularity between the definitions of reducible collection
continuations and reducible collections, the former are defined by lifting a
reducible term $M$ of type $T$ to a singleton collection.

%When we use this technique
In \NRCl, besides commuting conversions, we come across an additional problem concerning the property of distributivity of comprehensions over unions, represented by the following reduction rule:
\[
\comprehension{M \ocup N | x \gets P} \red \comprehension{M | x \gets P} \ocup \comprehension{N | x \gets P}
\]
One can immediately see that in $\comprehension{M \ocup N | x \gets \Box}$ the reduction above duplicates the hole, producing a multi-hole context that is not a continuation in the Lindley-Stark sense.

Cooper, in his work, attempted to reconcile continuations with
duplicating reductions. While considering extensions to his language,
we discovered that his proof of strong normalization presents a
nontrivial lacuna which we could only fix by relaxing the definition
of continuations to allow multiple holes. This problem affected both
the proof of the original result and our attempt to extend it, and has
an avalanche effect on definitions and proofs, yielding a more radical
revision of the $\top\top$-lifting technique which is the subject of
this paper.

The contribution of this paper is to place previous work on
higher-order programming for language-integrated query on a solid
foundation.  As we will show, our approach also extends to proving
normalization for a higher-order heterogeneous collection calculus
\NRClsb~\cite{ricciotti19dbpl} and we believe our proof technique
can be extended further.

This article is a revised and expanded version of a conference paper~\cite{Ricciotti20fscd}. Compared with the conference paper, this
article refines the notion of $\top\top$-lifting by omitting a harmless, but unnecessary generalization, includes details of proofs that had to be left out, and expands the discussion of related work. In addition, we fully comment on
the extension of our result to a language allowing to freely mix and compose set queries and bag queries, which was only marginally
discussed in the conference version. We also solved a subtle problem with the treatment of variable capture in contexts by reformulating the statement of Lemma~\ref{lem:redtoredctxapp}.

\subsection{Summary}
Section~\ref{sec:calculus} reviews $NRC_\lambda$ and its
rewrite system.  Section~\ref{sec:red} presents the refined approach
to reducibility needed to handle rewrite rules with branching
continuations.  Section~\ref{sec:sn} presents the proof of strong
normalization for $NRC_\lambda$.  Section~\ref{sec:heterogeneous}
outlines the extension to a higher-order calculus \NRClsb providing
heterogeneous set and bag queries.  Sections~\ref{sec:related}
and~\ref{sec:concl} discuss related work and conclude.

\section{Higher-order NRC}\label{sec:calculus}
\NRCl, a nested relational calculus with non-recursive higher-order functions, is defined by the following grammar:
\[
\begin{array}{lrcl}
\textbf{types} &
  S, T & ::= & A \orelse S \to T \orelse \tuple{\vect{\ell : T}} \orelse \setlit{T}
%  \orelse \msetlit{T}
\\
\textbf{terms} &
  L, M, N & ::= & x \orelse c(\vect{M}) \orelse \tuple{\vect{\ell = M}} \orelse M.\ell
       \orelse \lambda x.M \orelse (M~N) \\
       & & \orelse & \emptyoset \orelse \setlit{M} \orelse M \ocup N \orelse \comprehension{M | x \leftarrow N} \\
       & & \orelse & \plempty~M \orelse \plwhere~M~\kwdo~N
\end{array}
\]
where $x$, $\ell$ and $c$ range over countably infinite and pairwise disjoint sets of variables, record field labels, and constants.

Types include atomic types $A, B,\ldots $ (among which we have Booleans $\boolty$), record types with named fields $\tuple{\vect{\ell : T}}$, collections $\setlit{T}$; we define \emph{relation types} as those in the form $\setlit{\tuple{\vect{\ell:A}}}$, i.e.\ collections of records of atomic types.
Terms include applied constants $c(\vect{M})$, records with
named fields and record projections ($\tuple{\ell = M}$, $M.\ell$),
various collection terms (empty, singleton, union, and comprehension),
the emptiness test $\plempty$, and one-sided conditional expressions
for collection types
$\plwhere~M~\kwdo~N$; we allow the type of records with no fields: consisting of a single, empty record $\tuple{}$. Notice that $\lambda x.M$ and $\comprehension{M \mid x \gets N}$
bind the variable $x$ in $M$.

We will allow ourselves to use sequences of generators in comprehensions, which are syntactic sugar for nested comprehensions, e.g.:
\[
\comprehension{ M | x \leftarrow N, y \leftarrow R } \triangleq
\comprehension{ \comprehension{ M | y \leftarrow R } | x \leftarrow N }
\]

The typing rules, shown in Figure~\ref{fig:typing}, are largely
standard, and we only mention those operators that are specific to our
language: constants are typed according to a fixed signature $\Sigma$,
prescribing the types of the $n$ arguments and of the returned
expression to be atomic; we assume that $\Sigma$ assigns the type $\boolty$ to the constants $\kwtrue$ and $\kwfalse$ (representing the two Boolean values), and the type $(\boolty,\boolty) \to \boolty$ to the constant
$\land$ (representing the logical `and`) and we will allow ourselves to write $M \land N$ instead of $\land(M,N)$.
The operation $\isempty$ takes a collection and returns a
Boolean indicating whether its argument is empty; $\plwhere$ takes a
Boolean condition and a collection and returns the second argument if
the Boolean is true, otherwise the empty set.  (Conventional two-way
conditionals, at any type, are omitted for convenience but can be added without difficulty.)
\begin{figure}[tb]
\begin{center}
\AxiomC{$x : T \in \Gamma$}
\UnaryInfC{$\Gamma \vdash x : T$}
\DisplayProof
\hspace{.5cm}
\AxiomC{$\Sigma(c) = \vect{A_n} \imp A'$}
\AxiomC{$(\Gamma \vdash M_i : A_i)_{i = 1,\ldots,n}$}
\BinaryInfC{$\Gamma \vdash c(\vect{M_n}) : A'$}
\DisplayProof

\medskip

\AxiomC{$(\Gamma \vdash M_i : T_i)_{i = 1, \ldots, n}$}
\UnaryInfC{$\Gamma \vdash \tuple{\vect{\ell_n = M_n}} : \tuple{\vect{\ell_n : T_n}}$}
\DisplayProof
\hspace{.5cm}
\AxiomC{$\Gamma \vdash M : \tuple{\vect{\ell_n : T_n}}$}
\AxiomC{$i \in \metaset{1,\ldots,n}$}
\BinaryInfC{$\Gamma \vdash M.\ell_i : T_i$}
\DisplayProof

\medskip

\AxiomC{$\Gamma, x : S \vdash M : T$}
\UnaryInfC{$\Gamma \vdash \lambda x.M : S \to T$}
\DisplayProof
\hspace{.5cm}
\AxiomC{$\Gamma \vdash M : S \to T$}
\AxiomC{$\Gamma \vdash N : S$}
\BinaryInfC{$\Gamma \vdash (M~N) : T$}
\DisplayProof

\medskip

\AxiomC{$\phantom{A}$}
\UnaryInfC{$\Gamma \vdash \emptyoset : \setlit{T}$}
\DisplayProof
\hspace{.5cm}
\AxiomC{$\Gamma \vdash M : T$}
\UnaryInfC{$\Gamma \vdash \setlit{M} : \setlit{T}$}
\DisplayProof
\hspace{.5cm}
\AxiomC{$\Gamma \vdash M : \setlit{T}$}
\AxiomC{$\Gamma \vdash N : \setlit{T}$}
\BinaryInfC{$\Gamma \vdash M \ocup N : \setlit{T}$}
\DisplayProof

\medskip

\AxiomC{$\Gamma, x:T \vdash M : \setlit{S}$}
\AxiomC{$\Gamma \vdash N : \setlit{T}$}
\BinaryInfC{$\Gamma \vdash \comprehension{M | x \leftarrow N} : \setlit{S}$}
\DisplayProof

\medskip

\AxiomC{$\Gamma \vdash M : \setlit{T}$}
\UnaryInfC{$\Gamma \vdash \isempty~M : \boolty$}
\DisplayProof
\hspace{.5cm}
\AxiomC{$\Gamma \vdash M : \boolty$}
\AxiomC{$\Gamma \vdash N : \setlit{T}$}
\BinaryInfC{$\Gamma \vdash \plwhere~M~\kwdo~N : \setlit{T}$}
\DisplayProof
\end{center}
\caption{Type system of \NRCl.}\label{fig:typing}
\end{figure}

Overall, our presentation of \NRCl is very similar to the language of queries used by Cooper in~\cite{Cooper09a}: two minor differences are that \NRCl does not have a specific construct for input tables (these can be simulated as free variables with relation type) and uses one-armed conditionals instead of an if-then-else construct. Additionally, Cooper provided a type-and-effect system to track the use of primitive operations that may not be translated to SQL;\@ the issue of translating to SQL is not addressed directly in this paper (equivalently, we may assume that all the primitive operations of \NRCl may be translated to SQL).

\subsection{Reduction and normalization}
\NRCl is equipped with a rewrite relation $\red$ whose purpose is to convert expressions of relation type into a sublanguage isomorphic to a fragment of SQL, even when the original expression contains subterms whose type is not available in SQL, such as nested collections. This rewrite relation is obtained from the basic contraction $\ured$ shown in Figure~\ref{fig:normNRCl}, by taking its congruence closure (Figure~\ref{fig:congruenceNRCl}).

\begin{figure}[hbt]
\begin{center}\small
\begin{tabular}{r@{$~\ured~$}l}
\multicolumn{2}{c}{$
(\lambda x.M)~N \ured M[N/x]
\qquad
\tuple{\ldots, \ell = M, \ldots}.\ell \ured M
\qquad
c(c'_1, \ldots, c'_n) \ured \sem{c}(c'_1, \ldots, c'_n)
$}
\\
\multicolumn{2}{c}{} %% BLANK
\\
\multicolumn{2}{c}{$
\comprehension{\emptyoset | x \leftarrow M} \ured \emptyoset
\qquad
\comprehension{M | x \leftarrow \emptyoset} \ured \emptyoset
\qquad
\comprehension{M | x \leftarrow \setlit{N}} \ured M[N/x]
$}
\\
$\comprehension{M \ocup N | x \leftarrow R}$
 & $\comprehension{M | x \leftarrow R} \ocup \comprehension{N | x \leftarrow R}$
\\
$\comprehension{M | x \leftarrow N \ocup R}$
& $\comprehension{M | x \leftarrow N} \ocup \comprehension{M | x \leftarrow R}$
\\
$\comprehension{M | y \leftarrow \comprehension{R | x \leftarrow N }}$
& $\comprehension{ M | x \leftarrow N, y \leftarrow R} \qquad \text{(if $x \notin \FV(M)$)}$
\\
$\comprehension{M | x \leftarrow \plwhere~N~\kwdo~R}$
& $\plwhere~N~\kwdo~\comprehension{M | x \leftarrow R}$
\\
\multicolumn{2}{c}{} %% BLANK
\\
\multicolumn{2}{c}{$
\plwhere~\kwtrue~\kwdo~M \ured M
\qquad
\plwhere~\kwfalse~\kwdo~M \ured \emptyoset
\qquad
\plwhere~M~\kwdo~\emptyoset \ured \emptyoset
$}
\\
$\plwhere~M~\kwdo~(N \ocup R)$ & $(\plwhere~M~\kwdo~N) \ocup (\plwhere~M~\kwdo~R)$
\\
$\plwhere~M~\kwdo~\comprehension{N | x \leftarrow R}$
& $\comprehension{\plwhere~M~\kwdo~N | x \leftarrow R} \qquad \text{(if $x \notin \FV(M)$)}$
\\
$\plwhere~M~\kwdo~\plwhere~N~\kwdo~R$ & $\plwhere~(M \land N)~\kwdo~R$
\\
$\isempty~M$ & $\isempty~(\comprehension{\tuple{} | x \gets M}) \qquad \text{(if $M$ is not relation-typed)}$
\end{tabular}
\end{center}
\caption{Query normalization (basic contraction rules)}\label{fig:normNRCl}
\end{figure}
\begin{figure}[hbt]
\begin{center}
\AxiomC{$M \ured M'$}
\UnaryInfC{$M \red M'$}
\DisplayProof
\hspace{.5cm}
\AxiomC{$M \red M'$}
\UnaryInfC{$c(\vect{L},M,\vect{N}) \red c(\vect{L},M',\vect{N})$}
\DisplayProof

\medskip

\AxiomC{$M \red M'$}
\UnaryInfC{$\tuple{\vect{\ell_1 = L},\ell = M,\vect{\ell_2 = N}} \red \tuple{\vect{\ell_1 = L},\ell = M',\vect{\ell_2 = N}}$}
\DisplayProof
\hspace{.5cm}
\AxiomC{$M \red M'$}
\UnaryInfC{$M.\ell \red M'.\ell$}
\DisplayProof

\medskip

\AxiomC{$M \red M'$}
\UnaryInfC{$\lambda x.M \red \lambda x.M'$}
\DisplayProof
\hspace{.5cm}
\AxiomC{$M \red M'$}
\UnaryInfC{$M~N \red M'~N$}
\DisplayProof
\hspace{.5cm}
\AxiomC{$M \red M'$}
\UnaryInfC{$L~M \red L~M'$}
\DisplayProof

\medskip

\AxiomC{$M \red M'$}
\UnaryInfC{$\setlit{M} \red \setlit{M'}$}
\DisplayProof
\hspace{.5cm}
\AxiomC{$M \red M'$}
\UnaryInfC{$M \ocup N \red M' \ocup N$}
\DisplayProof
\hspace{.5cm}
\AxiomC{$M \red M'$}
\UnaryInfC{$L \ocup M \red L \ocup M'$}
\DisplayProof

\medskip

\AxiomC{$M \red M'$}
\UnaryInfC{$\comprehension{M | x \leftarrow N} \red \comprehension{M' | x \leftarrow N}$}
\DisplayProof
\hspace{.5cm}
\AxiomC{$M \red M'$}
\UnaryInfC{$\comprehension{L | x \leftarrow M} \red \comprehension{L | x \leftarrow M'}$}
\DisplayProof

\medskip

\AxiomC{$M \red M'$}
\UnaryInfC{$\isempty~M \red \isempty~M'$}
\DisplayProof
\medskip

\AxiomC{$M \red M'$}
\UnaryInfC{$\plwhere~M~\kwdo~N \red \plwhere~M'~\kwdo~N$}
\DisplayProof
\hspace{.5cm}
\AxiomC{$M \red M'$}
\UnaryInfC{$\plwhere~L~\kwdo~M \red \plwhere~L~\kwdo~M'$}
\DisplayProof
\end{center}
\caption{Query normalization (congruence closure of $\ured$)}\label{fig:congruenceNRCl}
\end{figure}
We will allow ourselves to say ``induction on the derivation of reduction'' to mean the structural induction induced by the notion of congruence closure, followed by a case analysis on the basic reduction rules used as its base case.

We now discuss the basic reduction rules in more detail. 0-ary constants are values of atomic type and do not reduce. Applied constants (with positive arity) reduce when all of their arguments
are (0-ary) constants: the reduction rule relies on a fixed semantics $\sem{\cdot}$ which assigns to each constant $c$ of signature $\Sigma(c) = \vect{A'_n} \imp A$ a function mapping constants $c'_1,\ldots,c'_n$ of type $\vect{A'_n}$ to values of type $A$. The rules for collections and conditionals are mostly standard. The reduction rule for the emptiness test is triggered when the argument $M$ is not of relation type (but, for instance, of nested collection type) and employs comprehension to generate a (trivial) relation that is empty if and only if $M$ is.

The normal forms of queries under these rewriting rules construct no
intermediate nested structures, and are straightforward to translate
to syntactically isomorphic (up to notational differences) and semantically equivalent SQL queries. For example, consider the following \NRCl query which, given a table $t$, first wraps the $id$ field of every tuple of $t$ into a singleton, yielding a collection of singletons (i.e.\ a nested collection), then converts it back to a flat collection by performing the grand union of all of its elements:
\[
\comprehension{y \mid y \gets \comprehension{ \setlit{\setlit{x.id}} \mid x \gets t}}
\]
The normal form of this query does not create the unnecessary intermediate nested collection:
\[
\comprehension{\setlit{x.id} \mid x \gets t}
\]
Such a query is easily translated to SQL as:
\[
\kwsel~x.id~\kwfrom~t~x
\]
Cooper~\cite{Cooper09} and Lindley and
Cheney~\cite{lindley12tldi} give a full account of the translation from \NRCl normal forms to SQL\@. Cheney
et al.~\cite{cheney13icfp} showed how to improve the performance and
reliability of LINQ in F\# using normalization and gave many examples
showing how higher-order queries support a convenient, compositional
language-integrated query programming style.

\section{Reducibility with branching continuations}%
\label{sec:red}
We introduce here the extension of $\top\top$-lifting we use to derive a proof of strong normalization for \NRCl. The main contribution of this section is a refined definition of continuations with branching structure and multiple holes, as opposed to the linear structure with a single hole used by standard $\top\top$-lifting. In our definition, continuations (as well as the more general notion of context) are particular forms of terms: in this way, the notion of term reduction can be used for continuations as well, without need for auxiliary definitions.

\subsection{Contexts and continuations}
We start our discussion by introducing \emph{contexts}, or terms with multiple, labelled holes that can be instantiated by plugging other terms (including other contexts) into them.

\begin{defi}[context]\label{def:context}
Let us fix a countably infinite set $\cP$ of indices: a \emph{context} $C$ is a term that may contain distinguished free variables $\ibox{p}$, also called \emph{holes}, where $p \in \cP$. Holes are never bound by any of the binders (we disallow terms of the form $\lambda \ibox{p}.M$ or $\comprehension{M \mid \ibox{p} \gets N}$).

Given a finite map from indices to terms $[p_1 \mapsto M_1, \ldots, p_n \mapsto M_n]$ (\emph{context instantiation}), the notation $C[p_1 \mapsto M_1, \ldots, p_n \mapsto M_n]$ (\emph{context application}) denotes the term obtained by simultaneously plugging $M_1,\ldots,M_n$ into the holes $\ibox{p_1},\ldots,\ibox{p_n}$. Notice that the $M_i$ are allowed to contain holes.

We will use metavariables $\eta, \theta$ to denote context instantiations.
\end{defi}

\begin{defi}[support]
Given a context $C$, its \emph{support} $\supp(C)$ is defined as the set of the indices
$p$ such that $\ibox{p}$ occurs in $C$:
\[
\supp(C) \triangleq \metaset{ p : \ibox{p} \in \FV(C) }
\]
\end{defi}
(it suffices to use $\FV(C)$ as holes are never used as bound variables).

When a term does not contain any $\ibox{p}$, we say that it is a \emph{pure} term; when it is important that a term be pure, we will refer to it by using overlined metavariables $\pure{L}, \pure{M}, \pure{N}, \pure{R}, \ldots$

We introduce a notion of \emph{permutable} multiple context instantiation.

\begin{defi}
A context instantiation $\eta$ is \emph{permutable} iff for all $p \in \dom(\eta)$ we have $\FV(\eta(p)) \cap \dom(\eta) = \emptyset$.
\end{defi}
The word ``permutable'' is explained by the following properties:
\begin{lem}\label{lem:permutable}
	Let $\eta$ be permutable, and $p_1,\ldots,p_k$ be an enumeration of all the elements of $\dom(\eta)$ without repetitions, in any order.
	Then, for all contexts $C$, we have:
	\[
		C\eta = C[p_1 \mapsto \eta(p_1)]\cdots[p_k \mapsto \eta(p_k)]
	\]
\end{lem}
\begin{proof}
	By structural induction on $C$. The relevant case is when $C = \ibox{p_i}$, for some $i \in \metaset{1,...,k}$. By definition, the left-hand side
	rewrites to $\eta(p_i)$; we can express the right-hand side as
	$\ibox{p_i}\vect{[p_j \mapsto \eta(p_j)]}_{j=1,\ldots,i-1}[p_i \mapsto \eta(p_i)]\vect{[p_j \mapsto \eta(p_j)]}_{j=i+1,\ldots,k}$. Then we prove:
	\begin{itemize}
		\item $\ibox{p_i}\vect{[p_j \mapsto \eta(p_j)]}_{j=1,\ldots,i-1} = \ibox{p_i}$, because for all $j$, $p_i \neq p_j$;
		\item $\ibox{p_i}[p_i \mapsto \eta(p_i)] = \eta(p_i)$ by definition;
		\item $\eta(p_i)\vect{[p_j \mapsto \eta(p_j)]}_{j=i+1,\ldots,k} = \eta(p_i)$ because, by the permutability hypothesis, for all $j$ we have
			that $\ibox{p_j} \notin \FV(\eta(p_i))$.
	\end{itemize}
	Then the right-hand side also rewrites to $\eta(p_i)$, proving the thesis.

	All the other cases are trivial, applying induction hypotheses where needed.
\end{proof}
\begin{lem}
Let $\eta$ be permutable and let us denote by $\eta_{\lnot p}$ the restriction of $\eta$ to indices other than $p$. Then for all $p \in \dom(\eta)$ we have:
\[
C \eta = C[p \mapsto \eta(p)]\eta_{\lnot p} = C\eta_{\lnot p}[p \mapsto \eta(p)]
\]
\end{lem}
\begin{proof}
	Immediate, by Lemma~\ref{lem:permutable}.
\end{proof}

We can now define continuations as certain contexts that capture how one or more collections can be used in a program.

\begin{defi}[continuation]\label{def:cont}
Continuations $K$ are defined as the following subset of contexts:
\begin{align*}
K,H ::= & \quad \ibox{p} \orelse \pure{M} \orelse K \ocup K
\orelse \comprehension{\pure{M}|x \gets K} \orelse \plwhere~\pure{B}~\kwdo~K
\end{align*}
where for all indices $p$, $\ibox{p}$ can occur at most once.
\end{defi}

This definition differs from the traditional one in two ways: first, holes are decorated with an index; secondly, and most importantly, the production $K \ocup K$ allows continuations to branch and, as a consequence, to use more than one hole. Note that the grammar above is ambiguous, in the sense that certain expressions like $\plwhere~\pure{B}~\kwdo~\pure{N}$ can be obtained either from the production $\plwhere~\pure{B}~\kwdo~K$ with $K = \pure{N}$, or as pure terms by means of the production $\pure{M}$: we resolve this ambiguity by parsing these expressions as pure terms whenever possible, and as continuations only when they are proper continuations.

An additional complication of \NRCl when compared to the computational metalanguage for which $\top\top$-lifting was devised lies in the way conditional expressions can reduce when placed in an arbitrary context: continuations in the grammar above are not liberal enough to adapt to such reductions, therefore, like Cooper, we will need an additional definition of \emph{auxiliary} continuations allowing holes to appear in the body of a comprehension (in addition to comprehension generators).
\begin{defi}[auxiliary continuation]\label{def:qcont}
Auxiliary continuations are defined as the following subset of contexts:
\begin{align*}
Q,O ::= & \quad \ibox{p} \orelse \pure{M} \orelse Q \ocup Q \orelse \comprehension{Q|x \gets Q}
\orelse \plwhere~\pure{B}~\kwdo~Q
\end{align*}
where for all indices $p$, $\ibox{p}$ can occur at most once.
\end{defi}
We can then see that regular continuations are a special case of auxiliary continuations;
however, an auxiliary continuation is allowed to branch not only with unions, but also with comprehensions.\footnote{It is worth noting that Cooper's original definition of auxiliary continuation does not use branching comprehension (nor branching unions), but is linear just like the original definition of continuation. The only difference between regular and auxiliary continuations in his work is that the latter allowed nesting not just within comprehension generators, but also within comprehension bodies (in our notation, this would correspond to two separate productions $\comprehension{\pure{M} | x \gets Q}$ and $\comprehension{Q|x \gets \pure{N}}$).}

We use the following definition of \emph{frames} to represent certain continuations with a distinguished shallow hole denoted by $\Box$.
\begin{defi}[frame]\label{def:frame}
Frames are defined by the following grammar:
\begin{align*}
F ::= & \quad \comprehension{Q|x \gets \Box} \orelse \comprehension{\Box | x \gets Q} \orelse \plwhere~\pure{B}~\kwdo~\Box
\end{align*}
where $\Box$ does not occur in $Q$, and for all indices $p$, $\ibox{p}$ can occur in $Q$ at most once.

The operation $F^p$, lifting a frame to an auxiliary continuation with a distinguished hole $\ibox{p}$ is defined as:
\[
F^p := F[ \Box \mapsto \ibox{p} ] \qquad (p \notin \supp(F))
\]

The composition operation $Q \compop{p} F$ is defined as:
\[
Q \compop{p} F = Q[p \mapsto F^p]
\]
\end{defi}
We generally use frames in conjunction with continuations or auxiliary continuations when we need to partially expose their leaves: for example, if we write $K = K_0 \compop{p} \comprehension{\pure{M} | x \gets \Box}$, we know that instantiating $K$ at index $p$ with (for example) a singleton term will create a redex: $K[p \mapsto \setlit{\pure{L}}] \red K_0[p \mapsto \pure{M}\subst{x}{\pure{L}}]$. We say that such a reduction is a \emph{reduction at the interface} between the continuation and the instantiation (we will make this notion formal in Lemma~\ref{lem:classification}).

In certain proofs by induction that make use of continuations, we will need to use a \emph{measure} of continuations to show that the induction is well-founded. We introduce here two measures $\len{\cdot}_p$ and $\sz{\cdot}_p$ denoting the nesting depth of a hole $\ibox{p}$: the two measures differ in the treatment of nesting within the body of a comprehension.
\begin{defi}\label{def:measures}
The measures $\len{Q}_p$ and $\sz{Q}_p$ are defined as follows:
\[
\begin{array}{c}
\len{\ibox{q}}_p = \sz{\ibox{q}}_p = \left\{
 \begin{array}{ll}
 1 & \mbox{if $p = q$} \\
 0 & \mbox{else}
 \end{array}
 \right.
\\
\len{\pure{M}}_p = \sz{\pure{M}}_p = 0
\\
\len{Q_1 \ocup Q_2}_p = \max(\len{Q_1}_p,\len{Q_2}_p)
  \qquad
\sz{Q_1 \ocup Q_2}_p = \max(\sz{Q_1}_p,\sz{Q_2}_p)
\\
\len{\plwhere~B~Q}_p = \len{Q}_p + 1
  \qquad
\sz{\plwhere~B~Q}_p = \sz{Q}_p + 1
\\
\len{\comprehension{Q_1|x \mapsto Q_2}}_p = \left\{
  \begin{array}{ll}
  \goldbox{$\len{Q_1}_p$} & \mbox{if $p \in \supp(Q_1)$} \\
  \len{Q_2}_p + 1 & \mbox{if $p \in \supp(Q_2)$} \\
  0 & \mbox{else}
  \end{array}
  \right.
\\
\sz{\comprehension{Q_1|x \mapsto Q_2}}_p = \left\{
  \begin{array}{ll}
  \goldbox{$\sz{Q_1}_p + 1$} & \mbox{if $p \in \supp(Q_1)$} \\
  \sz{Q_2}_p + 1 & \mbox{if $p \in \supp(Q_2)$} \\
  0 & \mbox{else}
  \end{array}
  \right.
\end{array}
\]
We will also use $\len{Q}$ and $\sz{Q}$ to refer to the derived measures:
\[
\len{Q} = \sum_{p \in \supp(Q)} \len{Q}_p
\qquad \qquad
\sz{Q} = \sum_{p \in \supp(Q)} \sz{Q}_p
\]
\end{defi}

The definitions of frames and measures are designed in such a way that the following property holds.
\begin{lem}\label{lem:szdecrease}
Let $Q$ be an auxiliary continuation such that $p \in \supp(Q)$; then for all frames $F$:
\begin{enumerate}
\item $\sz{Q}_p < \sz{Q \compop{p} F}_p$ and $\sz{Q} < \sz{Q \compop{p} F}$
\item if $F$ is not of the form $\comprehension{\Box \mid x \gets O}$, then $\len{Q}_p < \len{Q \compop{p} F}_p$ and $\len{Q} < \len{Q \compop{p} F}$
\item if $F = \comprehension{\Box \mid x \gets O}$, then $\len{Q}_p = \len{Q \compop{p} F}_p$ and $\len{Q} = \len{Q \compop{p} F}$.
\end{enumerate}
\end{lem}
\begin{proof}
By induction on the structure of $Q$. When examining the forms $Q$ can assume, we will have to consider subexpressions $Q'$ for which $p$ may or may not be in $\supp(Q')$: in the first case, we can apply the induction hypothesis; otherwise, we prove $\sz{Q'}_p = \sz{Q' \compop{p} F} = 0$ and $\len{Q'}_p = \len{Q' \compop{p} F} = 0$.
\end{proof}

\NRCl reduction can be used immediately on contexts (including regular and auxiliary continuations) since these are simply terms with distinguished free variables; we will also abuse notation to allow ourselves to specify reduction on context instantiations: whenever $\eta(p) \red N$ and $\eta' = \eta_{\lnot p}[p \mapsto N]$, we can write $\eta \red \eta'$.

We will denote the set of strongly normalizing terms by $\SN$. Strongly-normalizing applied contexts satisfy the following property:

For strongly normalizing terms (and by extension for context instantiations containing only strongly normalizing terms), we can introduce the concept of maximal reduction length.

\begin{defi}[maximal reduction length]
Let $M \in \SN$: we define $\maxred{M}$ as the maximum length of all reduction sequences starting with $M$. We also define $\maxred{\eta}$ as $\sum_{p \in \dom(\eta)} \maxred{\eta(p)}$, whenever all the terms in the codomain of $\eta$ are strongly normalizing.
\end{defi}

Since each term can only have a finite number of contracta, it is easy to see that $\maxred{M}$ is defined for any
strongly normalizing term $M$. Furthermore, $\maxred{M}$ is strictly decreasing under reduction.

\begin{lem}\label{lem:redmaxred}
For all strongly normalizing terms $M$, if $M \red M'$, then $\maxred{M'} < \maxred{M}$.
\end{lem}
\begin{proof}
If $\maxred{M'} \geq \maxred{M}$, by pre-composing $M \red M'$ with a reduction chain of maximal length starting at $M'$ we obtain a new reduction chain starting at $M$ with length strictly greater than $\maxred{M}$; this contradicts the definition of $\maxred{M}$.
\end{proof}

\subsection{Renaming reduction}
According to the Definitions~\ref{def:cont} and~\ref{def:qcont}, in order for a context to be a continuation or an auxiliary continuation, it must on one hand agree with the respective grammar, and on the other hand satisfy the condition that no hole occurs more than once. We immediately see that, since holes can be duplicated under reduction, the sets of plain and auxiliary continuations are not closed under reduction. For instance:
\[
K = \comprehension{M \cup N \mid x \gets \ibox{p}} \red \comprehension{M \mid x \gets\ibox{p}} \cup \comprehension{N \mid x \gets\ibox{p}} = C
\]
where $K$ is a continuation, but $C$ is not due to the two occurrences of $\ibox{p}$.
For this reason, we introduce a refined notion of renaming reduction which we can use to rename holes in the results so that each of them occurs at most one time.

\begin{defi}\label{def:renaming}
Given a term $M$ with holes and a finite map $\sigma : \cP \to \cP$, we write $M\sigma$ for the term obtained from $M$ by replacing each hole $\ibox{p}$ such that $\sigma(p)$ is defined with $\ibox{\sigma(p)}$.
\end{defi}

Even though finite renaming maps are partial functions, it is
convenient to extend them to total functions by taking $\sigma(p) = p$
whenever $p \notin \dom(\sigma)$; we will write $\id$ to denote the empty renaming map, whose total extension is the identity function on $\cP$.

\begin{defi}[renaming reduction]
$M$ $\sigma$-reduces to $N$ (notation: $M \ared{\sigma} N$) iff $M \red N\sigma$.
\end{defi}

Terms only admit a finite number of redexes and consequently, under regular reduction, any given term has a finite number of possible contracta. However, under renaming reduction, infinite contracta are possible: if
$M \red N$, there may be infinite $R, \sigma$ such that
$N = R \sigma$. When a strongly normalizing term $M$ admits infinite contracta, it does not necessarily have a maximal reduction sequence (just like the maximum of an infinite set of finite numbers is not necessarily defined). Fortunately, we can prove (Lemma~\ref{lem:redared}) that to every renaming
reduction chain there corresponds a plain reduction chain of the same
length, and vice-versa.

\begin{lem}\label{lem:redrename}
If $M \red N$, then for all $\sigma$ we have $M \sigma \red N \sigma$.
\end{lem}
\begin{proof}
Routine induction on the derivation of $M \red N$.
\end{proof}
\begin{lem}\label{lem:redared}\mbox{}\\
For every finite plain reduction sequence, there is a corresponding renaming reduction sequence of the same length (using the identity renaming $\id$); and conversely, for every finite renaming reduction sequence, there is a corresponding plain reduction sequence of the same length involving renamed terms. More precisely:
\begin{enumerate}
\item If $M_0 \red \cdots \red M_n$, then $M_0 \ared{\id} \cdots \ared{\id} M_n$
\item If $M_0 \ared{\sigma_1} \cdots \ared{\sigma_{n-1}} M_{n-1} \ared{\sigma_n} M_n$, then $M_0 \red \cdots M_{n-1}\sigma_{n-1}\cdots\sigma_1 \red M_n \sigma_n\cdots\sigma_1$
\end{enumerate}
\end{lem}
\begin{proof}
The first part of the lemma is trivial. For the second part, proceed by induction on the length of the reduction chain: in the inductive case, we have $M_0 \ared{\sigma_1}\cdots\ared{\sigma_n} M_n \ared{\sigma_{n+1}} M_{n+1}$ by hypothesis and $M_0 \red \cdots \red M_n\sigma_n\cdots\sigma_1$ by induction hypothesis; to obtain the thesis, we only need to prove that
\[
M_n\sigma_n\cdots\sigma_1 \red M_{n+1}\sigma_{n+1}\sigma_n\cdots\sigma_1
\]
In order for this to be true, by Lemma~\ref{lem:redrename}, it is sufficient to show that $M_n \red M_{n+1} \sigma_{n+1}$; this is by definition equivalent to $M_n \ared{\sigma_{n+1}} M_{n+1}$, which we know by hypothesis.
\end{proof}
\begin{cor}\label{cor:aredmaxred}
Suppose $M \in \SN$: if $M \ared{\sigma} M'$, then $\maxred{M'}$ is defined and $\maxred{M'} < \maxred{M}$.
\end{cor}
\begin{proof}
By Lemma~\ref{lem:redared}, for any plain reduction chain there exists a renaming reduction chain of the same length, and vice-versa. Thus, since plain reduction lowers the length of the maximal reduction chain (Lemma~\ref{lem:redmaxred}), the same holds for renaming reduction.
\end{proof}

The results above prove that the set of strongly normalizing
terms is the same under the two notions of reduction, thus $\maxred{M}$
can be used to refer to the maximal length of reduction chains starting at $M$
either with or without renaming.

Our goal is to describe the reduction of pure terms expressed in the form of applied continuations. One first difficulty we need to overcome is that, as we noted, the sets of continuations (both regular and auxiliary) are not closed under reduction: the duplication of holes performed by reduction will produce contexts that are not continuations or auxiliary continuations because they do not satisfy the condition of the linearity of holes. Thankfully, renaming reduction allows us to restore the linearity of holes, as we show in the following lemma.
\begin{lem}\label{lem:contred}\mbox{}
\begin{enumerate}
\item%
\label{case:contred_inK}
For all continuations $K$, if $K \red C$, there exist a continuation $K'$ and a finite map $\sigma$ such that $K \ared{\sigma} K'$ and $K'\sigma = C$.
\item%
\label{case:contred_inQ}
For all auxiliary continuations $Q$, if $Q \red C$, there exist an auxiliary continuation $Q'$ and a finite map $\sigma$ such that $Q \ared{\sigma} Q'$ and $Q'\sigma = C$.
\end{enumerate}
Furthermore, the $\sigma, K', Q'$ in the statements above can be chosen so that $\dom(\sigma)$ is fresh with respect to any given finite set of indices $\cS$.
\end{lem}
\begin{proof}
Let $\cS$ be a finite set of indices and $C$ a contractum of the continuation we wish to reduce. This contractum will not, in general, satisfy the linearity condition of holes that is mandated by the definitions of plain and auxiliary continuations; however we can show that, for any context with duplicated holes, there exists a structurally equal context with linear holes.

Operationally, if $C$ contains $n$ holes, we generate $n$ different indices that are fresh for $\cS$, and replace the index of each hole in $C$ with a different fresh index to obtain a new context $C'$: this induces a finite map $\sigma : \supp(C') \to \supp(C)$ such that $C'\sigma = C$. By the definition of renaming reduction, we have $K \ared{\sigma} C'$ (resp. $Q \ared{\sigma} C'$). To prove that $C'$ is a continuation (resp.\ auxiliary continuation) we need to show that it satisfies the linearity condition and that it meets the grammar in Definition~\ref{def:cont} (resp. Definition~\ref{def:qcont}). The first part holds by construction; the proof that $C'$ satisfies the required  grammar is obtained by structural induction on the derivation of the reduction, with a case analysis on the structure of $K$ (or on the structure of $Q$).

By construction of $\sigma$, we also have that $\dom(\sigma) \cap \cS = \emptyset$, as required by the Lemma statement.
\end{proof}

A further problem concerns variable capture: if we reduce $C \red C'$, there is no guarantee that $C\eta \red C'\eta$ for a given context instantiation $\eta$. This happens for two reasons: the first one is that reducing a context $C$ may cause a hole to move within the scope of a new binder. So,
\[
\begin{array}{c}
\comprehension{\ibox{p} \mid y \gets \comprehension{N \mid x \gets M}} \red \comprehension{\ibox{p} \mid x \gets M, y \gets N}
\\
\text{but}
\\
\comprehension{\ibox{p} \mid y \gets \comprehension{N \mid x \gets M}}[p \mapsto x] \not\red \comprehension{\ibox{p} \mid x \gets M, y \gets N}[p \mapsto x]
\end{array}
\]
because the first term is equal to $\comprehension{x \mid y \gets \comprehension{N \mid x \gets M}}$ where the $x$ in the head of the outer comprehension is free, and the reduction is blocked until we rename the bound $x$ of the inner comprehension.

The second reason for which the reduction of a context may be disallowed if we apply it to a context instantiation is that, due to variable capture, the reduction may involve the context instantiation in a non-trivial way:
\[
\begin{array}{c}
(\lambda z.\ibox{p})~N \red \ibox{p}
\\
\text{but}
\\
((\lambda z.\ibox{p})~N)[p \mapsto z] \not\red \ibox{p}[p \mapsto z]
\end{array}
\]
because the left-hand term is equal to $(\lambda z.\ibox{p}[p \mapsto z])~N$, and the right-hand one is $z$, and the former does not reduce to the latter, but to $\ibox{p}[p \mapsto z]\subst{z}{N} = N$.

While we understand that the first of the two problems should be handled with a suitable alpha-renaming of the redex, the other is more complicated. Fortunately, in most cases we are not interested in the reduction of generic contexts, but only in that of auxiliary continuations: due to their restricted term shape, auxiliary continuations only allow some reductions, most of which do not present the problem above; the exception is when a reduction is obtained by contracting a subterm using the comprehension-singleton rule:
\[
\comprehension{Q_0 \mid z \gets \setlit{\pure{L}}} \red Q_0 \subst{z}{\pure{L}}
\]
By applying a context instantiation $\eta$ to both sides, we obtain an incorrect contraction:
\[
\comprehension{Q_0 \mid z \gets \setlit{\pure{L}}}\eta = \comprehension{Q_0\eta \mid z \gets \setlit{\pure{L}}} \not\red Q_0\subst{z}{\pure{L}}\eta
\]
where the left-hand term does not reduce to the right-hand one because in the latter the codomain of $\eta$ might contain free instances of $z$ that have not been replaced by $\pure{L}$.
In the rest of the paper, we will call reductions using the comprehension-singleton rule \emph{special reductions}. When $Q \red Q'$ by means of a special reduction, we know that in general $Q\eta \not\red Q'\eta$ (not even after alpha-renaming), and we will have to handle such a case differently.

If however $Q \red Q'$ is \emph{not} a special reduction, to mimic that reduction within $Q\eta$ we may start by renaming the bound variables of this term in such a way that no reduction is blocked: we obtain a term in the form $O\theta$, where $O$ is an auxiliary continuation alpha-equivalent\footnote{
In our setting, contexts are defined as a particular case of terms, allowing special ``hole'' free variables that are not used in binders; thus, we only have a single notion of alpha-equivalence for terms that we also apply to contexts (just like our notion of reduction works on terms and contexts alike). This may look surprising and perhaps suspicious, considering that in some formal treatments of contexts (e.g.~\cite{bognar2001}) the alpha-renaming of contexts is forbidden; however, our work does not need to provide a general treatment of alpha-renaming in contexts: we only use it under special conditions that ensure its consistency.
} to $Q$, and $\theta$ is obtained from $\eta$ by replacing some of its free variables with other free variables, consistently with the renaming of $Q$ to $O$, as shown by the following lemma (a more general result applying to all reductions $C \red C'$ and all context instantiations $\eta$, which however gives weaker guarantees on the result of contracting $C\eta$, will be provided as Lemma~\ref{lem:redrename_aux4}).
\begin{lem}\label{lem:redtoredctxapp}
For all auxiliary continuations $Q$ and for all permutable context instantiations $\eta$, there exist an auxiliary continuation $O$ and a context instantiation $\theta$ such that:
\begin{enumerate}
\item $Q =_\alpha O$ and $Q\eta =_\alpha O\theta$
\item $\theta = [ p \mapsto \eta(p)\subst{\vect{x_p}}{\vect{y_p}} \mid p \in \dom(\eta)]$ for some $\vect{x_p}, \vect{y_p}$ (i.e.\ $\theta$ is equal to $\eta$ up to a renaming of the free variables in its codomain)
\item for all $C'$ such that $Q \red C'$ with a non-special reduction, there exists $D' =_\alpha C'$ such that $O \red D'$ and we also have $O\theta \red D'\theta$
\end{enumerate}
\end{lem}
\begin{proof}
  We proceed by induction on the size of $Q$ followed by a case analysis on its structure; for each case, after considering all possible reductions starting in that particular shape of $Q$ (where we are allowed, by the hypothesis, to ignore special reductions), we perform a renaming of $Q\eta$ to $O\theta$ that is guaranteed to allow us to prove the thesis. Particular care is needed when context instantiations cross binders, as variable capture is allowed to happen:
  \begin{itemize}
  \item Case $Q = \ibox{p}$: no reduction of $Q$ is possible, so we can choose $O := Q$, $\theta := \eta$, and the thesis holds trivially.
  \item Case $Q = \pure{M}$: for all reductions $\pure{M} \red \pure{M'}$, we have $\pure{M}\eta = \pure{M}$ and $\pure{M'}\eta = \pure{M'}$, because context instantiation is ineffective on pure terms; so we can choose $O := Q$, $\theta := \eta$, and the thesis holds trivially.
  \item Case $Q = Q_1 \cup Q_2$. Since the holes in $Q$ are linear and $\eta$ is permutable, we can decompose $\eta = \eta_1\eta_2$, such that $Q\eta = Q_1\eta_1 \cup Q_2\eta_2$; we apply the induction hypothesis twice on the two subterms, to obtain $O_1, O_2, \theta_1, \theta_2$ such that for $i=1,2$, we have $O_i =_\alpha Q_i$, $O_i\theta_i =_\alpha Q_i\eta_i$, $\theta_i$ is equal to $\eta_i$ up to a renaming of the free variables in its codomain, and for all $C'_i$ such that $Q_i \red C'_i$ where the reduction is not special, there exists $D'_i =_\alpha C'_i$ such that $O_i \red D'_i$ and $O_i\theta_i \red D'_i\theta_i$.
  Now, to prove the thesis, we fix $O := O_1 \cup O_2$ and $\theta := \theta_1\theta_2$; we easily show that $O =_\alpha Q$, $O\theta =_\alpha Q\eta$ and that $\theta$ is equal to $\eta$ up to a renaming of the free variables in its codomain. To conclude the proof, we consider any given reduction $Q \red C'$, and we see by case analysis that either $C' = Q_1' \cup Q_2$ such that $Q_1 \red Q_1'$, or $C' = Q_1 \cup Q_2'$ such that $Q_2 \red Q_2'$: in the first case, we know that there exists $D'_1 =_\alpha C'_1$ such that $O_1 \red D'_1$ and $O_1\theta_1 \red D'_1$: we fix $D' := D'_1 \cup C'_2$ and easily prove $D' =_\alpha C'$, $D'\theta =_\alpha C'\eta$, and $O\theta \red D'\theta$.
  \item Case $Q = \comprehension{Q_1 \mid y \gets \comprehension{Q_2 \mid z \gets Q_3}}$.
  Since the holes in $Q$ are linear and $\eta$ is permutable, we can decompose $\eta = \eta_1\eta_2\eta_3$, such that \[Q\eta = \comprehension{Q_1\eta_1 \mid y \gets \comprehension{Q_2\eta_2 \mid z \gets Q_3\eta_3}};\]
  $Q$ can reduce either in the subcontinuations $Q_1$, $Q_2$, $Q_3$ or, if $z \notin \FV(Q_1)$, by applying the unnesting reduction; however, the last reduction might be blocked in $Q\eta$, if $z \in \FV(Q_1\eta_1)$. For this reason, we start by choosing a non-hole variable $z^* \notin \FV(Q_1\eta_1)$ and renaming $Q$ as $Q^* := \comprehension{Q_1 \mid y \gets \comprehension{Q^*_2 \mid z^* \gets Q_3}}$, where we have defined $Q^*_2 := Q_2 \subst{z}{z^*}$. Clearly, $Q^* =_\alpha Q$ and, if we fix $\eta^*_2 := [\eta_2(p)\subst{z}{z^*} \mid p \in \dom(\eta)]$ and $\eta^* := \eta_1\eta^*_2\eta_3$, we also have $Q^*\eta^* =_\alpha Q\eta$; furthermore, since $z^*$ is not a hole, we can see $\eta^*$ must still permutable.
  Now, since $Q^*\eta^* = \comprehension{Q_1\eta_1 \mid y \gets \comprehension{Q_2^*\eta_2^* \mid z^* \gets Q_3\eta_3}}$,
  we apply the induction hypothesis three times on the subterms, to obtain:
  \begin{itemize}
  \item $O_1, \theta_1$ such that $O_1 =_\alpha Q_1$, $O_1\theta_1 =_\alpha Q_1\eta_1$, $\theta_1$ is equal to $\eta_1$ up to a renaming of the free variables in its codomain, and for all $C'_1$ such that $Q_1 \red C'_1$ where the reduction is not special, there exists $D'_1 =_\alpha C'_1$ such that $O_1 \red D'_1$ and $O_1\theta_1 \red D'_1\theta_1$
  \item $O_2, \theta_2$ such that $O_2 =_\alpha Q^*_2$, $O_2\theta_2 =_\alpha Q^*_2\eta^*_2$, $\theta_1$ is equal to $\eta^*_2$ (and thus to $\eta_2$) up to a renaming of the free variables in its codomain, and for all $C'_2$ such that $Q^*_2 \red C'_2$ where the reduction is not special, there exists $D'_2 =_\alpha C'_2$ such that $O_2 \red D'_2$ and $O_2\theta_2 \red D'_2\theta_2$
  \item $O_3, \theta_3$ such that $O_3 =_\alpha Q_3$, $O_3\theta_3 =_\alpha Q_3\eta_3$, $\theta_3$ is equal to $\eta_3$ up to a renaming of the free variables in its codomain, and for all $C'_3$ such that $Q_3 \red C'_3$ where the reduction is not special, there exists $D'_3 =_\alpha C'_3$ such that $O_3 \red D'_3$ and $O_3\theta_3 \red D'_3\theta_3$
  \end{itemize}
  Note that the first and third case are similar, but the second one has slightly different properties to account for the renaming of $z$ to $z^*$.

  Now we fix $O := \comprehension{O_1 \mid y \gets \comprehension{O_2 \mid z^* \gets O_3}}$ and $\theta := \theta_1\theta_2\theta_3$. We easily show that $O =_\alpha Q$, $O\theta =_\alpha Q\eta$ and that $\theta$ is equal to $\eta$ up to a renaming of the free variables in its codomain. To conclude the proof, we consider any given reduction $Q \red C'$, and we see by case analysis that either the reduction corresponds to a reduction in the subcontinuations $Q_i$ (in which case we conclude by a reasoning on the subterms and induction hypotheses, similarly to the union case above), or to one of the following:
  \begin{itemize}
  \item $C' = \comprehension{Q_1 \mid z \gets Q_3, y \gets Q_2}$: in this case, we fix $D' := \comprehension{Q_1 \mid z^* \gets Q_3, y \gets Q^*_2}$ and we prove, as required, that $D' =_\alpha C'$, $O \red D'$, and $O\theta \red D'\theta$.
  \item $Q_1 = \emptyset$ and $C' = \emptyset$: we easily see that $O_1 = \emptyset$, so we can fix $D' = \emptyset$ and show the thesis.
  \item $Q_1 = Q_{11} \cup Q_{12}$ and
  \[
    C' = \comprehension{Q_{11} \mid y \gets \comprehension{Q_2 \mid z \gets Q_3}} \cup \comprehension{Q_{12} \mid y \gets \comprehension{Q_2 \mid z \gets Q_3}}:\]
  we prove that there exist $O_{11}, O_{12}$ such that $Q_{11} =_\alpha O_{11}$, $Q_{12} =_\alpha O_{12}$, and $O_1 = O_{11} \cup O_{12}$, fix $D' = \comprehension{Q_{11} \mid y \gets \comprehension{Q_2 \mid z \gets Q_3}} \cup \comprehension{Q_{12} \mid y \gets \comprehension{Q_2 \mid z \gets Q_3}}$, and prove the thesis.
  \end{itemize}
  \item Case $Q = \comprehension{Q_1 \mid y \gets Q_2}$ where $Q_2$ is not a comprehension: this is similar to the case above, but instead of comprehension unnesting we have to consider two possible reductions
  \begin{itemize}
  \item $Q_2 = \plwhere~\pure{B}~\kwdo~Q_3$ and $C' = \plwhere~\pure{B}~\kwdo~\comprehension{Q_1 \mid y \gets Q_3}$
  \item $Q_2 = Q_3 \cup Q_4$ and $C' = \comprehension{Q_1 \mid y \gets Q_3} \cup \comprehension{Q_1 \mid y \gets Q_4}$
  \end{itemize}
  However, these reductions do not require us to perform renamings and do not pose any problems.
  \item Case $Q = \plwhere~\pure{B}~\kwdo~Q_1$. Besides reductions in the subterms, we have to consider the following cases:
  \begin{itemize}
  \item $Q_1 = \comprehension{Q_2 \mid y \gets Q_3}$ and $C' = \comprehension{\plwhere~\pure{B}~\kwdo~Q_1 \mid z \gets Q_2}$ where $z \notin \FV(\pure{B})$
  \item $\pure{B} = \kwtrue$ and $C' = Q_1$
  \item $\pure{B} = \kwfalse$ and $C' = \emptyset$
  \item $Q_1 = Q_2 \cup Q_3$ and $C' = (\plwhere~\pure{B}~\kwdo~Q_2) \cup (\plwhere~\pure{B}~\kwdo~Q_3)$
  \item $Q_1 = \plwhere~\pure{B_0}~\kwdo~Q_2$ and $C' = \plwhere~(B \land B_0)~\kwdo~Q_2$.
  \end{itemize}
  In all these cases, no renaming is required (besides those produced by applying the induction hypothesis); in particular, the first reduction is always possible without renaming because $\pure{B}$ is a pure term, so $z \notin \FV(\pure{B}\eta)$ because $\pure{B}\eta = \pure{B}$. Therefore, we prove the thesis using the induction hypothesis and an exhaustive case analysis on the possible reduction as we did above, without particular problems.
	 \qedhere
  \end{itemize}
\end{proof}

\begin{rem}
It is important to understand that, unlike all other operations on terms, context instantiation is not defined on the abstract syntax, independently of the particular choice of names, but on the concrete syntax. In other words, all operations and proofs that do not use context instantiation work on alpha-equivalence classes of terms; but when context instantiation is used, say on a context $C$, we need to choose a representative of the alpha-equivalence class of $C$.

Thanks to Lemma~\ref{lem:redtoredctxapp}, whenever we need to reduce $Q$ with a non-special reduction in a term of the form $Q\eta$, we may assume without loss of generality that the representative of the alpha-equivalence class of $Q\eta$ is chosen so that if $Q \red C'$, then $Q\eta \red C'\eta$. Technically, we prove that there exist $O, D', \theta$ such that $Q =_\alpha O$, $C' =_\alpha D'$, $\theta$ is equal to $\eta$ up to renaming, and $O\theta \red D'\theta$, but after the context instantiation is completed, we return to consider terms as equal up to alpha-equivalence. The result will be used in the proof of the following Lemma~\ref{lem:ctxred_ctxappred}, where we have clarified the technical parts.
\end{rem}

Finally, given a non-special renaming reduction $Q \ared{\sigma} Q'$, we want to be able to express the corresponding reduction on $Q\eta$: due to the renaming $\sigma$, it is not enough to change $Q$ to $Q'$, but we also need to construct some $\eta'$ containing precisely those mappings $[q \mapsto M]$ such that, if $\sigma(q) = p$, then $p \in \dom(\eta)$ and $\eta(p) = M$. This construction is expressed by means of the following operation.
\begin{defi}\label{def:ren_inst}
For all context instantiations $\eta$ and renamings $\sigma$,
%
%let $\overline{\sigma}$ be the extension of $\sigma$ to the domain of $\eta$ such that $\overline{\sigma}(p) = \sigma(p)$ if $p \in \dom(\sigma)$, and $\overline{\sigma}(p) = p$ if $p \notin \dom(\sigma)$ and $p \in \dom(\eta)$. Then
%
we define $\eta^\sigma$ as the context instantiation such that:
\begin{itemize}
\item if $\sigma(p) \in \dom(\eta)$ then $\eta^\sigma(p) = \eta(\sigma(p))$;
\item in all other cases, $\eta^\sigma(p) = p$.
\end{itemize}
\end{defi}

\noindent
The results above allow us to express what happens when a reduction duplicates the holes in a continuation which is then combined with a context instantiation.

\begin{lem}\label{lem:renamectxapp}
  For all contexts $C$, finite maps $\sigma$, and context instantiations $\eta$ such that, for all $p \in \dom(\eta)$, $\supp(\eta(p)) \cap \dom(\sigma) = \emptyset$, we have $C\sigma\eta = C\eta^\sigma\sigma$.
\end{lem}
\begin{proof}
  By structural induction on $C$. The interesting case is when $C = \ibox{p}$. If $\sigma(p) \in \dom(\eta)$:
  \begin{align*}
  \ibox{p}\sigma\eta & = \ibox{\sigma(p)}\eta & & \text{(Definition~\ref{def:renaming})}
  \\
  & = \eta(\sigma(p)) & & \text{(Definition~\ref{def:context})}
  \\
  & = \eta(\sigma(p))\sigma & & \text{($\supp(\eta(p)) \cap \dom(\sigma) = \emptyset$)}
  \\
  & = \ibox{p}\eta^\sigma\sigma & & \text{(Definition~\ref{def:ren_inst}, with $\sigma(p) \in \dom(\eta)$)}
  \end{align*}
  If instead $\sigma(p) \notin \dom(\eta)$:
  \begin{align*}
  \ibox{p}\sigma\eta & = \ibox{\sigma(p)}\eta & & \text{(Definition~\ref{def:renaming})}
  \\
  & = \ibox{\sigma(p)} & & \text{($\sigma(p) \notin \dom(\eta)$)}
  \\
  & = \ibox{p}\sigma & & \text{(Definition~\ref{def:renaming})}
  \\
  & = \ibox{p}\eta^\sigma\sigma & & \text{(Definition~\ref{def:ren_inst}, with $\sigma(p) \notin \dom(\eta)$)} \qedhere
  \end{align*}
\end{proof}
\begin{lem}\label{lem:ctxred_ctxappred}
    For all auxiliary continuations $Q$, renamings $\sigma$, and permutable context instantiations $\eta$ such that, for all $p \in \dom(\eta)$, $\supp(\eta(p)) \cap \dom(\sigma) = \emptyset$, there exist an auxiliary continuation $O$ and a context instantiation $\theta$ such that:
	\begin{enumerate}
	\item $O =_\alpha Q$ and $O\theta =_\alpha Q\eta$
	\item $\theta = [ p \mapsto \eta(p)\subst{\vect{x_p}}{\vect{y_p}} \mid p \in \dom(\eta)]$ for some $\vect{x_p}, \vect{y_p}$ (i.e.\ $\theta$ is equal to $\eta$ up to a renaming of the free variables in its codomain)
	\item for all $Q'$ such that $Q \ared{\sigma} Q'$ with a non-special reduction, there exists $O' =_\alpha Q'$ such that $O \ared{\sigma} O'$ and we also have $O\theta \red O'\theta$
	\end{enumerate}
\end{lem}
\begin{proof}
    By Lemma~\ref{lem:renamectxapp}, we obtain $O$ and $\theta$ such that the alpha-equivalences and the condition on $\theta$ hold. Furthermore, by the definition of $\ared{\sigma}$, we have $O \red O'\sigma$; then, again by Lemma~\ref{lem:redtoredctxapp}, we obtain $O\theta \red O'\sigma\theta$; by Lemma~\ref{lem:renamectxapp}, we know $O'\sigma\eta = O'\eta^\sigma\sigma$; then the thesis $O\theta \ared{\sigma} O'\eta^\sigma$ follows immediately by the definition of $\ared{\sigma}$.
\end{proof}

\begin{rem}
In~\cite{Cooper09a}, Cooper attempts to prove strong normalization for \NRCl using a similar, but weaker result:
\begin{quote}
If $K \red C$, then for all terms $M$ there exists $K'_M$ such that $C[M] = K'_M[M]$ and $K[M] \red K'_M[M]$.
\end{quote}
Since he does not have branching continuations and renaming reductions, whenever a hole is duplicated, e.g.
\[
K = \comprehension{N_1 \cup N_2 | x \gets \Box} \red \comprehension{N_1 | x \gets \Box} \cup \comprehension{N_2 | x \gets \Box} = C
\]
he resorts to obtaining a continuation from $C$ simply by filling one of the holes with the term $M$:
\[
K'_M = \comprehension{N_1 | x \gets M} \cup \comprehension{N_2 | x \gets \Box}
\]
Hence, $K'_M[M] = C[M]$. Unfortunately, subsequent proofs rely on the fact that $\maxred{K}$ must decrease under reduction: since we have no control over $\maxred{M}$, which could potentially be much greater than $\maxred{K}$, it may be that $\maxred{K'_M} \geq \maxred{K}$.

In our setting, by combining Lemmas~\ref{lem:contred} and~\ref{lem:ctxred_ctxappred}, we can find a $K'$ which is a proper contractum of $K$. By Lemma~\ref{lem:redmaxred}, we get $\maxred{K'} < \maxred{K}$, as required by subsequent proofs.
\end{rem}
%\smallskip

More generally, the following lemma will help us in performing case analysis on the reduction of an applied continuation.
\begin{lem}[classification of reductions in applied continuations]\label{lem:classification}
  Suppose $Q\eta \red N$, where $\eta$ is permutable, and $\dom(\eta) \subseteq \supp(Q)$; then one of the following holds:
  \begin{enumerate}
  \item\label{case:incont}
  there exist an auxiliary continuation $Q'$ and a finite map $\sigma$ such that $N = Q'\eta^\sigma$, $Q \ared{\sigma} Q'$ and, for all $p \in \dom(\eta)$, $\supp(\eta(p)) \cap \dom(\sigma) = \emptyset$: we say that this is a reduction of the continuation $Q$;
  \item\label{case:tricky}
  there exist auxiliary continuations $Q_1, Q_2$, an index $q \in \supp(Q_1)$, a variable $x$, and a term $L$ such that $Q = (Q_1 \compop{q} \comprehension{\Box \mid x \gets \setlit{\pure{L}}})[q \mapsto Q_2]$, $N = Q'\eta^*$, and $Q \red Q'$, where we define $Q' = Q_1[q \mapsto Q_2\subst{x}{\pure{L}}]$ and $\eta^*(p) = \eta(p)\subst{x}{\pure{L}}$ for all $p \in \supp(Q_2)$, otherwise $\eta^*(p) = \eta(p)$: this is a special reduction of the continuation $Q$;
  \item there exists a permutable $\eta'$ such that $N = Q\eta'$ and $\eta \red \eta'$: in this case we say the reduction is \emph{within $\eta$};
  \item there exist an auxiliary continuation $Q_0$, an index $p$ such that $p \in \supp(Q_0)$ and $p \in \dom(\eta)$, a frame $F$ and a term $M$ such that $N = Q_0[p \mapsto M]\eta_{\lnot p}$, $Q = Q_0 \compop{p} F$, and $F^p[p \mapsto \eta(p)] \red M$: in this case we say the reduction is \emph{at the interface}.
  \end{enumerate}
  Furthermore, if $Q$ is a regular continuation $K$, then the  $Q'$ in case~\ref{case:incont} can be chosen to be a regular continuation $K'$, and case~\ref{case:tricky} cannot happen.
  \end{lem}
  \begin{proof}
  By induction on $Q$ with a case analysis on the reduction rule applied. In case~\ref{case:incont}, to satisfy the property relating $\eta$ and $\sigma$, we use Lemma~\ref{lem:contred} to generate a $\sigma$ such that the indices of its domain are fresh with respect to the codomain of $\eta$. To see that this partition of reductions is exhaustive, the most difficult part is to check that whenever we are in the case of a reduction at the interface, there is a suitable $F$ such that $Q$ can be decomposed as $Q_0 \compop{p} F$; while there are some reduction rules for which we cannot find a suitable $F$, the structure of $Q$ prevents these from happening at the interface between $Q$ and $\eta$: for example, in a reduction $(Q_0 \compop{p} (\Box~L))[p \mapsto \lambda x.M] = Q_0[p \mapsto (\lambda x.M)~L] \red Q_0[p \mapsto M\subst{x}{L}]$, $(\Box~L)$ is not a valid frame: but we do not have to consider this case, because $Q$ cannot be of the form $Q_0 \compop{p} (\Box~L)$, since the latter is not a valid auxiliary continuation.
  \end{proof}

For all context instantiations $\eta$, case~\ref{case:incont} of the Lemma above generates a renaming $\sigma$ satisfying the hypotheses of Lemma~\ref{lem:renamectxapp} and Lemma~\ref{lem:ctxred_ctxappred}. Additionally, the following result states that $\eta^\sigma$ must be permutable.
\begin{lem}
Suppose that for all $p \in \dom(\eta)$, $\supp(\eta(p)) \cap \dom(\sigma) = \emptyset$. Then, $\eta^\sigma$ is permutable.
\end{lem}
\begin{proof}
We need to prove that, for all $p \in \dom(\eta^\sigma)$, $\supp(\eta^\sigma(p)) \cap \dom(\sigma) = \emptyset$. If $p \in \dom(\eta^\sigma)$, then $\sigma(p) \in \dom(\eta)$; then, by hypothesis, we prove $\supp(\eta(\sigma(p))) \cap \dom(\sigma) = \emptyset$. Since $\eta(\sigma(p)) = \eta^\sigma(p)$, this proves the thesis. % TODO: unmatched brace?
\end{proof}

\begin{lem}\label{lem:redrename_aux1}
For all contexts $C$, context instantiations $\eta$, and sets of hole indices $\cS$, there exist a context $D$, a context instantiation $\theta$, and a hole renaming $\sigma$ such that:
\begin{itemize}
\item $D\sigma = C$ and $D\theta\sigma = C\eta$
\item the holes in $D$ are linear and $\forall \ibox{q} \in \FV(D)$, $q \notin \cS$
\item $\dom(\theta) \cup \dom(\sigma) \subseteq \FV(D)$
\item $\theta$ is permutable
\end{itemize}
\end{lem}
\begin{proof}[Proof sketch]
If $C$ has $n$ hole occurrences, we generate $n$ distinct indices $p_1,\ldots,p_n$ (which we take to be fresh with respect to $\cS$ and the free variables of the codomain of $\eta$)
and replace each hole occurrence within $C$ with a different $\ibox{p_i}$: this induces a context $D$ and a renaming $\sigma$ such that $D\sigma = C$.
By Lemma~\ref{lem:renamectxapp} we prove $C\eta = D\sigma\eta = D\eta^\sigma\sigma$ (we can apply this lemma thanks to the careful choice of the $p_i$). We take $\theta \triangleq \eta^\sigma$ and the remaining properties
follow easily (the permutability of $\theta$, again, descends from choosing sufficiently fresh indices $p_i$).
\end{proof}

\begin{lem}\label{lem:redrename_aux2}
Let $C_1,C_2$ be contexts and $\eta_1, \eta_2$ context instantiations. Then for all free variables $x$ and sets of hole indices $\cS$, there exist a context $D$, a permutable context instantiation $\theta$, and a hole renaming $\sigma$ such that:
\begin{itemize}
\item $D\sigma =_\alpha C_1\subst{x}{C_2}$ and $D\theta\sigma =_\alpha (C_1\eta_1)\subst{x}{C_2\eta_2}$
\item The holes in $D$ are linear and fresh with respect to $S$
\item For all $q \in \dom(\theta) \cup \dom(\sigma)$, $\ibox{q} \in \FV(D)$
\item $\theta$ is permutable
\end{itemize}
\end{lem}
\begin{proof}
The proof is by induction on the size of $C_1$, followed by a case analysis on its structure. Here we consider the variable cases, lambda as a template for binder cases, and application as a template for cases with multiple subterms.
\begin{itemize}
\item If $C_1 = x$, we have $x\subst{x}{C_2} = C_2$ and $(x\eta_1)\subst{x}{C_2\eta_2} = C_2\eta_2$. By Lemma~\ref{lem:redrename_aux2}, we find $D,\theta,\sigma$ such that $D\sigma = C_2$, $D\theta\sigma = C_2\eta_2$, the holes in $D$ are linear and arbitrarily fresh, for all $q \in \dom(\theta) \cup \dom(\sigma)$, $\ibox{q} \in \FV(D)$, and $\theta$ is permutable; this proves the thesis.
\item If $C_1 = \ibox{p} \in \dom(\eta_1)$, then $\ibox{p}\subst{x}{C_2} = \ibox{p}$ and $\ibox{p}\eta_1\subst{x}{C_2\eta_2} = \eta_1(p)\subst{x}{C_2\eta_2}$;
	for all sets of hole indices $\cS$, we choose $p^* \notin \cS$ such that $\ibox{p^*} \notin \FV(\eta_1(p)\subst{x}{C_2\eta_2})$; finally we choose $D = \ibox{p^*}$, $\theta = [p^* \mapsto \eta_1(p)\subst{x}{C_2\eta_2}]$, and $\sigma = [p^* \mapsto p]$ and prove the thesis.
\item If $C_1$ is a free variable $y$ not covered by the previous cases, we choose $D = y$, $\theta = []$ (the empty instantiation), $\sigma = []$ (the empty renaming) to trivially prove the thesis.
\item If $C_1 = \lambda y. C_0$, let us choose a variable $y^* \notin \metaset{x} \cup \FV(C_2\eta_2)$, such that $y^*$ is not a hole; let us define the following abbreviations:
\[
C_0^* \triangleq C_0\subst{y}{y^*}
\qquad
\eta_1^* \triangleq [p \mapsto \eta_1(p)\subst{y}{y^*} | p \in \dom(\eta_1)]
\]
Since $C_0^*$ is equal to $C_0$ up to a renaming, it is smaller than $C_1$ and by induction hypothesis we get that there exist $D_0, \theta_0, \sigma_0$ such that $D_0\sigma_0 =_\alpha C_0^*\subst{x}{C_2}$
and $D_0\theta_0\sigma_0 =_\alpha C_0^*\eta_1^*\subst{x}{C_2\eta_2}$, where the holes in $D_0$ are linear and arbitrarily fresh, for all $q \in \dom(\theta_0) \cup \dom(\sigma_0)$ we have $\ibox{q} \in \FV(D_0)$,
and $\theta_0$ is permutable. Then we can choose $D = \lambda y^*.D_0$, $\theta = \theta_0$, $\sigma = \sigma_0$ and show that:
\[
\begin{array}{rll}
D\sigma & = & (\lambda y^*.D_0)\sigma_0
\\
& = & \lambda y^*.D_0\sigma_0
\\
& =_\alpha & \lambda y^*.C_0^*\subst{x}{C_2}
\\
& = & (\lambda y.C_0)\subst{x}{C_2}
\\
D\theta\sigma & = & (\lambda y^*.D_0)\theta_0\sigma_0
\\
& = & \lambda y^*.D_0\theta_0\sigma_0
\\
& =_\alpha & \lambda y^*.C_0^*\eta_1^*\subst{x}{C_2\eta_2}
\\
& = & (\lambda y.C_0\eta_1)\subst{x}{C_2\eta_2}
\\
& = & (\lambda y.C_0)\eta_1\subst{x}{C_2\eta_2}
\end{array}
\]
We can easily show that the other required properties of $D, \theta, \sigma$ are verified, thus proving the thesis.
\item If $C_1 = (C_{11}~C_{12})$, then
\begin{mathpar}
 (C_{11}~C_{12})\subst{x}{C_2} = (C_{11}\subst{x}{C_2} C_{12}\subst{x}{C_2})
 \quad\text{and} \and
 (C_{11}~C_{12}\eta_1)\subst{x}{C_2\eta_2} = (C_{11}\eta_1\subst{x}{C_2\eta_2} C_{12}\eta_1\subst{x}{C_2\eta_2});
\end{mathpar}
by the induction hypothesis, we find $D_{11},\theta_{11},\sigma_{11}$ and $D_{12},\theta_{12},\sigma_{12}$ such that $D_{11}\sigma_{11} =_\alpha C_{11}\subst{x}{C_2}$, $D_{11}\theta_{11}\sigma_{11} =_\alpha C_{11}\eta_1\subst{x}{C_2\eta_2}$: we are allowed to choose these expressions in such a way that the holes in the $D_i$ are fresh, so we ensure that the sets of holes of $D_{11}$ and $D_{12}$ are disjoint -- this in turn ensures that $\theta_{11}, \theta_{12}$ and $\sigma_{11}, \sigma_{12}$ can be combined together, since their domains are also disjoint. Then we choose $D = (D_{11}~D_{12})$, $\theta = \theta_{11}\theta_{12}$, $\sigma = \sigma_{11}\sigma_{12}$ and show that;
\[
\begin{array}{rll}
D\sigma & = & (D_{11}~D_{12})\sigma_{11}\sigma_{12}
\\
& = & (D_{11}\sigma_{11}\sigma_{12}~D_{12}\sigma_{11}\sigma_{12})
\\
& = & (D_{11}\sigma_{11}~D_{12}\sigma_{12})
\\
& =_\alpha & (C_{11}\subst{x}{C_2}~C_{12}\subst{x}{C_2})
\\
D\theta\sigma & = & (D_{11}~D_{12})\theta_{11}\theta_{12}\sigma_{11}\sigma_{12}
\\
& = & (D_{11}\theta_{11}\theta_{12}\sigma_{11}\sigma_{12}~D_{12}\theta_{11}\theta_{12}\sigma_{11}\sigma_{12})
\\
& = & (D_{11}\theta_{11}\sigma_{11}~D_{12}\theta_{12}\sigma_{12})
\\
& =_\alpha & (C_{11}\eta_1\subst{x}{C_2\eta_2}~C_{12}\eta_1\subst{x}{C_2\eta_2})
\end{array}
\]
Furthermore, the induction hypothesis provides enough information on the $D_i,\theta_i,\sigma_i$ to guarantee that the holes in $D$ are linear and arbitrarily fresh,
for all $q \in \dom(\theta) \cup \dom(\sigma)$ we have $\ibox{q} \in \FV(D)$, and $\theta$ is permutable, as required.
\qedhere
\end{itemize}
\end{proof}

\begin{lem}\label{lem:redrename_aux4}
Let $C, C'$ be contexts such that $C \red C'$; then, for all context instantiations $\eta$, there exist a context $D$, context instantiation $\theta$, and renaming $\sigma$ such that
$D\sigma =_\alpha C'$ (and consequently $C \ared{\sigma} D$) and $C\eta \ared{\sigma} D\theta$.
\end{lem}
\begin{proof}
By structural induction on the derivation of $C \red C'$. The property we need to prove essentially states that any reduction of $C$ can still be performed after applying any instantiation $\eta$;
however, due to variable capture and the possibility that some redexes of $C$ may be blocked in $C\eta$, the statement is complicated by explicit alpha-conversions and hole renamings.
We present here the two interesting cases of the proof:
\begin{itemize}
\item $(\lambda x.C_1)~C_2 \red C_1\subst{x}{C_2}$: by Lemma~\ref{lem:redrename_aux2} we obtain $D,\theta,\sigma$ such that $D\sigma =_\alpha C_1\subst{x}{C_2}$ and $D\theta\sigma =_\alpha (C_1\eta)\subst{x}{C_2\eta}$; since $((\lambda x.C_1)~C_2)\eta = (\lambda x.C_1\eta)~(C_2\eta) \red (C_1\eta)\subst{x}{C_2\eta} =_\alpha D\theta\sigma$, we have $((\lambda x.C_1)~C_2)\eta \ared{\sigma} D\theta$.
\item $\comprehension{C_1 \mid x \gets \comprehension{C_2 \mid y \gets C_3}} \red \comprehension{C_1 \mid y \gets C_3, x \gets C_2}$, where $y \notin \FV(C_1)$: let us choose a variable $y^* \notin \FV(C_1\eta)$ and define $C^*_2 \triangleq C_2\subst{y}{y^*}$ and $\eta^* = [ p \mapsto \eta(p)\subst{y}{y^*} | p \in \dom(\eta)]$: then we can alpha-rename the contractum as
\[
\comprehension{C_1 \mid y \gets C_3, x \gets C_2} =_\alpha \comprehension{C_1 \mid y^* \gets C_3, x \gets C^*_2}
\]
By repeated applications of Lemma~\ref{lem:redrename_aux1}, we obtain contexts $C'_1,C'_2,C'_3$, context instantiations $\eta'_1,\eta'_2,\eta'_3$, and renamings $\sigma_1,\sigma_2,\sigma_3$, such that $\comprehension{C'_1 \mid y \gets C'_3, x \gets C'_2}$ has linear, arbitrarily fresh holes, for $i =1,2,3$ we have $\dom(\eta'_i)\cup \dom(\sigma_i) \subseteq \FV(C'_i)$ and $\eta'_i$ is permutable, and such that:
\[
\begin{array}{c}
C'_1\sigma_1 = C_1 \qquad C'_2\sigma_1 = C^*_2 \qquad \qquad C'_3\sigma_3 = C_3
\\
C'_1\eta'_1\sigma_1 = C_1\eta \qquad C'_2\eta'_2\sigma_2 = C^*_2\eta^* \qquad C'_3\eta'_3\sigma_3 = C_3\eta
\end{array}
\]
Then we take $D \triangleq \comprehension{C'_1 \mid y^* \gets C'_3, x \gets C'_2}$, $\theta = \eta'_1\eta'_2\eta'_3$ and $\sigma = \sigma_1\sigma_2\sigma_3$. We prove:
\[
\begin{array}{rll}
D\sigma & = & \comprehension{C'_1\sigma_1 \mid y^* \gets C'_3\sigma_3, x \gets C'_2\sigma_2}
\\
& = & \comprehension{C_1 \mid y^* \gets C_3, x \gets C^*_2}
\\
& =_\alpha & \comprehension{C_1 \mid y \gets C_3, x \gets C_2}
\\
D\theta\sigma & = & \comprehension{C'_1\eta'_1\sigma_1 \mid y^* \gets C'_3\eta'_3\sigma_3, x \gets C'_2\eta'_2\sigma_2}
\\
& = & \comprehension{C_1\eta \mid y^* \gets C_3\eta, x \gets C^*_2\eta^*}
\end{array}
\]
By the last equality, we have $\comprehension{C_1 \mid y^* \gets C_3, x \gets C^*_2} \red D\theta\sigma$, which by definition is $\comprehension{C_1 \mid y^* \gets C_3, x \gets C^*_2} \ared{\sigma} D\theta$.
\qedhere
\end{itemize}
\end{proof}

\noindent
The following result, like many others in the rest of this section, proceeds by well-founded induction; we will use the following notation to represent well-founded relations:
\begin{itemize}
\item $<$ stands for the standard less-than relation on $\bbN$, which is well-founded;
\item $\lessdot$ is the lexicographic extension of $<$ to $k$-tuples in $\bbN^k$ (for a given $k$), also well-founded;
\item $\prec$ will be used to provide a decreasing metric that depends on the specific proof: such metrics are defined as subsets of $\lessdot$ and are thus well-founded.
\end{itemize}

\begin{lem}\label{lem:maxredctxapp}\label{lem:ctxappsn}\label{lem:Kappsn}
Let $C$ be a context and $\eta$ a context instantiation such that $C\eta \in \SN$. Then we have:
\begin{enumerate}
\item $C \in \SN$
\item $\maxred{C} \leq \maxred{C\eta}$
\item for all $p \in \dom(\eta)$, if $\ibox{p} \in \FV(C)$, then $\eta(p) \in \SN$.
\end{enumerate}
\end{lem}
\begin{proof}
Property 3 follows immediately by induction on $\maxred{C\eta}$ by noticing that, since $\ibox{p} \in \FV(C)$ implies that $\eta(p)$ appears as a subexpression of $C\eta$, and since reduction is defined by congruence closure, every reduction of $\eta(p)$ can be mimicked by a corresponding reduction within $C\eta$.

To prove the first two properties, we proceed by well-founded induction on $(C,\eta)$ using the metric:
\[
(C_1,\eta_1) \prec (C_2,\eta_2) \iff \exists \sigma: C_2\eta_2 \ared{\sigma} C_1\eta_1
\]
We consider all the possible contractions $C \red C'$. By Lemma~\ref{lem:redrename_aux4}, we find $D,\sigma,\theta$ such that $C\eta \ared{\sigma} D\theta$ and $D\sigma =_\alpha C'$;
consequently, $C \ared{\sigma} D$. By induction hypothesis we obtain $D \in \SN$ and $\maxred{D} \leq \maxred{D\theta} < \maxred{C\eta}$; we easily prove $D\sigma =_\alpha C' \in \SN$ and $\maxred{C'} = \maxred{D}$.
Furthermore, since $\maxred{C} = 1 + \max_{C' : C \red C'}\maxred{C'}$ and for all such $C'$ we have proved $\maxred{C'} < \maxred{C\eta}$, we get $\maxred{C} \leq \maxred{C\eta}$.
\end{proof}

A similar property about the composition of continuations and frames follows immediately.
\begin{cor}\label{lem:maxredsubK}
If $Q \compop{p} F \in \SN$, then $Q \in \SN$ and $\maxred{Q} \leq \maxred{Q \compop{p} F}$.
\end{cor}
\begin{proof}
By definition, $Q \compop{p} F = Q[p \mapsto F^p]$: then we use Lemma~\ref{lem:maxredctxapp}, with $\eta = [p \mapsto F^p]$.
\end{proof}

\begin{lem}\label{lem:ctxappcomb}
Let $Q$ be an auxiliary continuation, and let $\eta, \theta$ be context instantiations s.t.\ their union is permutable. If $Q\eta \in \SN$ and $Q\theta \in \SN$, then $Q\eta\theta \in \SN$.
\end{lem}
\begin{proof}
We assume that $\dom(\eta) \cup \dom(\theta) \subseteq \supp(Q)$ (otherwise, we can find strictly smaller permutable $\eta', \theta'$ such that $Q\eta\theta = Q\eta'\theta'$, and their domains are subsets of $\supp(Q)$). We show $Q\eta \in \SN$ and $Q\theta \in \SN$ imply $Q \in \SN$, $\eta \in \SN$ and $\theta \in \SN$; thus we can then prove the theorem by well-founded induction on $(Q,\eta,\theta)$ using the following metric:
\[
(Q_1,\eta_1,\theta_1) \prec (Q_2,\eta_2,\theta_2)
\iff
(\maxred{Q_1},\sz{Q_1},\maxred{\eta_1}+\maxred{\theta_1})
\lessdot
(\maxred{Q_2},\sz{Q_2},\maxred{\eta_2}+\maxred{\theta_2})
\]
We show that all of the possible contracta of $Q\eta\theta$ are s.n.\ by case analysis on the contraction:
\begin{itemize}
\item $Q' \eta^\sigma \theta^\sigma$, where $Q \ared{\sigma} Q'$: it is easy to see that $\maxred{\eta^\sigma}$ and $\maxred{\theta^\sigma}$ are defined because $\maxred{\eta}$ and $\maxred{\theta}$ are; then the thesis follows from the induction hypothesis, knowing that $\maxred{Q'} < \maxred{Q}$ (Lemma~\ref{lem:redmaxred}).
\item $Q'\eta^*\theta^*$ where $Q'= Q_1[q \mapsto Q_2\subst{x}{\pure{L}}]$, for some $Q_1, Q_2, x, q \in \supp(Q_1), \eta^*, \theta^*$ such that $Q = (Q_1 \compop{q} \comprehension{\Box \mid x \gets \setlit{\pure{L}}})[q \mapsto Q_2]$ and $\eta^*(p) = \eta(p)\subst{x}{\pure{L}}$, $\theta^*(p) = \theta(p)\subst{x}{\pure{L}}$ for all $p \in \supp(Q_2)$, otherwise $\eta^*(p) = \eta(p)$ and $\theta^*(p) = \theta(p)$; since $Q \red Q'$, we know $\maxred{Q'} < \maxred{Q}$; furthermore, since $Q\eta \red Q'\eta^*$ and $Q\eta \in \SN$, it is easy to see that $Q'\eta^* \in \SN$ and $\maxred{\eta^*}$ is defined; similarly, $Q'\theta^* \in \SN$ and $\maxred{\theta^*}$ is defined; then $Q'\eta^*\theta^* \in \SN$ by induction hypothesis.
\item $Q\eta'\theta$, where $\eta \red \eta'$: the thesis follows by induction hypothesis, knowing $\maxred{\eta'} < \maxred{\eta}$ (Lemma~\ref{lem:redmaxred}).%$\item $Q\eta\theta'$, where $\theta \red \theta'$: similar to the case above.
\item $Q_0[p\mapsto N] \eta_0 \theta$ where $Q = Q_0 \compop{p} F$, $\eta = [p \mapsto M] \eta_0$, and $F^p[p \mapsto M] \red N$ by means of a reduction at the interface.
By Lemma~\ref{lem:maxredsubK} we know $\maxred{Q_0} \leq \maxred{Q}$; by Lemma~\ref{lem:szdecrease} we prove $\sz{Q_0} < \sz{Q}$. We take $\eta' = [p \mapsto N]\eta_0$: since $Q\eta$ reduces to $Q_0\eta'$ and both terms are strongly normalizing, we have that $\maxred{\eta'}$ is defined. Then we observe $(Q_0, \eta', \theta) \prec (Q, \eta, \theta)$ and obtain the thesis by induction hypothesis. A symmetric case with $p \in \dom(\theta)$ is proved similarly.
\qedhere
\end{itemize}
\end{proof}

\begin{cor}\label{cor:aredforall}
$Q[p \mapsto M]^\sigma \in \SN$ iff for all $q$ s.t.\ $\sigma(q) = p$, we have $Q[q \mapsto M] \in \SN$.
\end{cor}
\begin{proof}
By the definition of $[p \mapsto M]^\sigma$, using Lemma~\ref{lem:ctxappcomb} to decompose the resulting context instantiation.
\end{proof}

\subsection{Candidates of reducibility}
We here define the notion of \emph{candidates of reducibility}: sets of strongly normalizing terms enjoying certain closure properties that can be used to overapproximate the sets of terms of a certain type. Our version of candidates for \NRCl is a straightforward adaptation of the standard definition given by Girard and like that one is based on a notion of \emph{neutral terms}, i.e.\ those terms that, when placed in an arbitrary context, do not create additional redexes.

\begin{defi}[neutral term]\label{def:neutral}
A term $M$ is \emph{neutral} if it belongs to the following grammar:
\[
W \mathrel{::=}
x
\orelse c(\vect{M_n})
\orelse M.\ell
\orelse (M~N)
\orelse \isempty~M
\]
where $n \geq 1$.

The set of neutral terms is denoted by $\NT$.
\end{defi}

Let us introduce the following notation for Girard's CRx properties of sets~\cite{GLT89}:
\begin{itemize}
\item $\CRi(\cC) \triangleq \cC \subseteq \SN$
\item $\CRii(\cC) \triangleq \forall M \in \cC, \forall M'. M \red M' \then M' \in \cC$
\item $\CRiii(\cC) \triangleq \forall M \in \NT. (\forall M'. M \red M' \then M' \in \cC) \then M \in \cC$
\end{itemize}

\noindent
The set $\CR$ of the candidates of reducibility is then defined as the collection of those sets of terms which satisfy all the CRx properties. Some standard results include the non-emptiness of candidates (in particular, all free variables are in every candidate) and that $\SN \in \CR$.

\subsection{Reducibility sets}
In this section we introduce \emph{reducibility sets}, which are sets of terms that we will use to provide an interpretation of the types of \NRCl; we will then prove that reducibility sets are candidates of reducibility, hence they only contain strongly normalizing terms.
The following notation will be useful as a shorthand for certain operations on sets of terms that are used to define reducibility sets:
\begin{itemize}
\item $\cC \imp \cD \triangleq \metaset{M : \forall N \in \cC, (M~N) \in \cD}$
\item $\tuple{\vect{\ell_k : \cC_k}} \triangleq \metaset{M : \forall i = 1,\ldots,k, M.\ell_i \in \cC_i}$
\item $\topset{\cC_p} \triangleq \metaset{K : \forall M \in \cC. K[p \mapsto \setlit{M}] \in \SN}$
\item $\toptopset{\cC} \triangleq \metaset{M : \forall p, \forall K \in \topset{\cC_p}, K[p \mapsto M] \in \SN}$
\end{itemize}

\noindent
The sets $\topset{\cC_p}$ and $\toptopset{\cC}$ are called the \emph{$\top$-lifting} and \emph{$\top\top$-lifting} of $\cC$. These definitions refine the ones used in the literature by using indices: $\top$-lifting is defined with respect to a given index $p$, while the definition of $\top\top$-lifting uses any index (in the standard definitions, continuations only contain a single hole, and no indices are mentioned).

\begin{defi}[reducibility]
For all types $T$, the set $\Red{T}$ of reducible terms of type $T$ is defined by recursion on $T$ by means of the rules:
\[
\begin{array}{rclcrcl}
\Red{A} & \triangleq & \SN
& \qquad &
\Red{S \imp T} & \triangleq & \Red{S} \imp \Red{T}
\\
\Red{\tuple{\vect{\ell_k : T_k}}} & \triangleq & \tuple{\vect{\ell_k : \Red{T_k}}}
& \qquad &
\Red{\setlit{T}} & \triangleq & \toptopset{\Red{T}}
\end{array}
\]
\end{defi}

Let us use metavariables $\cS, \cS', \ldots$ to denote finite sets of indices: we provide a refined notion of $\top$-lifting $\topset{\cC_\cS}$ depending on a set of indices rather than a single index, defined by pointwise intersection. This notation is useful to track a $\top$-lifted candidate under renaming reduction.
\begin{defi}
$\topset{\cC_\cS} \triangleq \bigcap_{p \in \cS} \topset{\cC_p}$.
\end{defi}

\begin{defi}
Let $\cC$ and $\cS$ be sets of terms and indices respectively, and $\sigma$ a finite renaming: then we define $(\topset{\cC_\cS})^\sigma :=  \topset{\cC_{\sigma^{-1}(\cS)}}$, where $\sigma^{-1}(\cS) = \metaset{ q : \sigma(q) \in \cS }$
\end{defi}

We now proceed with the proof that all the sets $\Red{T}$ are candidates of reducibility: we will only focus on collections since for the other types the result is standard. The proofs of $\CRi$ and $\CRii$ do not differ much from the standard $\top\top$-lifting technique.

\begin{lem}\label{lem:setKbox}
Suppose $\CRi(\cC)$: then for all indices $p,q$, $\ibox{p} \in \topset{\cC_q}$.
\end{lem}
\begin{proof}
To prove the lemma, it is sufficient to show that for all $M \in \cC$ we have $\ibox{p}[q \mapsto \setlit{M}] \in \SN$. This term is equal to either $\setlit{M}$ (if $p = q$) or to $\ibox{p}$ (otherwise); both terms are s.n.\ (in the case of $\setlit{M}$, this is because $\CRi$ holds for $\cC$, thus $M \in \SN$).
\end{proof}
\begin{lem}[CR1 for continuations]\label{lem:cr1setK}
For all $p$ and all non-empty $\cC$, $\topset{\cC_p} \subseteq \SN$.
\end{lem}
\begin{proof}
We assume $K \in \topset{\cC_p}$ and $M \in \cC$: by definition, we know that $K[p \mapsto \setlit{M}] \in \SN$; then we have $K \in \SN$ by Lemma~\ref{lem:Kappsn}.
\end{proof}
\begin{lem}[CR1 for collections]\label{lem:cr1set}
If $\CRi(\cC)$, then $\CRi(\toptopset{\cC})$.
\end{lem}
\begin{proof}
We need to prove that if $M \in \toptopset{\cC}$, then $M \in \SN$. By the definition of $\toptopset{\cC}$, we know that for all $p$, $K[p \mapsto M] \in \SN$ whenever $K \in \topset{\cC_p}$. Now assume any $p$, and by Lemma~\ref{lem:setKbox} choose $K=\ibox{p}$: then $K[p \mapsto M] = M \in \SN$, which proves the thesis.
\end{proof}
\begin{lem}[CR2 for collections]\label{lem:cr2set}
If $M \in \toptopset{\cC}$ and $M \red M'$, then $M' \in \toptopset{\cC}$.
\end{lem}
\begin{proof}
Let $p$ be an index, and take $K \in \topset{\cC_p}$: we need to prove $K[p \mapsto M'] \in \SN$. By the definition of $M \in \toptopset{\cC}$, we have $K[p \mapsto M] \in \SN$; if $p \notin \supp(K)$, $K[p \mapsto M'] = K[p \mapsto M]$ and the thesis trivially holds; otherwise the instantiation is effective and we have $K[p \mapsto M] \red K[p \mapsto M']$, and this last term, being a contractum of a strongly normalizing term, is strongly normalizing as well. This proves the thesis.
\end{proof}

In order to prove $\CRii$ for all types (and particularly for collections), we do not need to establish an analogous property on continuations; however such a property is still useful for subsequent results (particularly $\CRiii$). Its statement must, of course, consider that reduction may duplicate (or indeed delete) holes, and thus employs renaming reduction. We can show that whenever we need to prove a statement about $n$-ary permutable instantiations of $n$-ary continuations, we can simply consider each hole separately, as stated in the following lemma.
\begin{lem}\label{lem:Kcrsimpl}
  $K \in (\topset{\cC_\cS})^\sigma$ if, and only if, for all $q \in \sigma^{-1}(\cS)$, we have $K \in \topset{\cC_q}$.
  \\
  In particular, $K \in (\topset{\cC_p})^\sigma$ if, and only if, for all $q$ s.t.\ $\sigma(q) = p$, we have $K \in \topset{\cC_q}$.
\end{lem}
\begin{proof}
  By definition of $(\cdot)^\sigma$ and $\topset{(\cdot)}$:
\[
	K \in (\topset{\cC}_\cS)^\sigma \iff K \in \topset{\cC}_{\sigma^{-1}(\cS)} \iff K \in \bigcap_{q \in \sigma^{-1}(\cS)} \topset{\cC}_q \iff \forall q \in \sigma^{-1}(\cS), K \in \topset{\cC}_q
    \qedhere
\]
\end{proof}

\begin{lem}[CR2 for continuations]\label{lem:cr2setK}
If $K \in \topset{\cC_\cS}$ and $K \ared{\sigma} K'$, then $K' \in (\topset{\cC_\cS})^\sigma$.
\end{lem}
\begin{proof}
By Lemma~\ref{lem:Kcrsimpl} and the definition of $\topset{(\cdot)}$, it suffices to prove that $K'[q \mapsto \setlit{M}] \in \SN$ for all $q$ such that $\sigma(q) \in \cS$ and $M \in \cC$.
Then we know $K[\sigma(q) \mapsto \setlit{M}] \in \SN$, and consequently $K'[\sigma(q) \mapsto \setlit{M}]^\sigma \in \SN$ as well, since the latter is a contractum of the former; finally, by Corollary~\ref{cor:aredforall}, $K'[q \mapsto \setlit{M}] \in \SN$, as we needed.
\end{proof}

This is everything we need to prove $\CRiii$.

\begin{lem}[CR3 for collections]\label{lem:cr3set}
Let $\cC \in \CR$, and $M$ a neutral term such that for all reductions $M \red M'$ we have $M' \in \toptopset{\cC}$. Then $M \in \toptopset{\cC}$.
\end{lem}
\begin{proof}
By definition, we need to prove $K[p \mapsto M] \in \SN$ whenever $K \in \topset{\cC_p}$ for some index $p$. By Lemma~\ref{lem:cr1setK}, knowing that $\cC$, being a candidate, is non-empty, we have $K \in \SN$. We can thus proceed by well-founded induction on $\maxred{K}$ to prove the strengthened statement: for all indices $q$, if $K \in \topset{\cC_q}$, then $K[q \mapsto M] \in \SN$.
Equivalently, we prove that all the contracta of $K[q \mapsto M]$ are s.n.\ by cases on the possible contracta:
\begin{itemize}
\item $K' [q \mapsto M]^\sigma$ (where $K \ared{\sigma} K'$): to prove this term is s.n., by Corollary~\ref{cor:aredforall}, we need to show $K'[q'\mapsto M] \in \SN$ whenever $\sigma(q') = q$; by Lemmas~\ref{lem:cr2setK} and~\ref{lem:Kcrsimpl}, we know $K' \in \topset{\cC_{q'}}$, and naturally $\maxred{K'} < \maxred{K}$ (Lemma~\ref{lem:redmaxred}), thus the thesis follows by the IH\@.
\item $K[p \mapsto M']$ (where $M \red M'$): this is s.n.\ because $M' \in \toptopset{\cC}$ by hypothesis.
\item Since $M$ is neutral, there are no reductions at the interface.
\qedhere
\end{itemize}
\end{proof}

\begin{thm}\label{thm:rediscr}
For all types $T$, $\Red{T} \in \CR$.
\end{thm}
\begin{proof}
Standard by induction on $T$. For $T = \setlit{T'}$, we use Lemmas~\ref{lem:cr1set},~\ref{lem:cr2set}, and~\ref{lem:cr3set}.
\end{proof}

\section{Strong normalization}%
\label{sec:sn}
We have proved that the reducibility sets of all types are candidates of reducibility. We are now going to prove that every well-typed term is in the reducibility set corresponding to its type: strong normalization will then follow as a corollary, by using the $\CRi$ property of candidates of reducibility.

The proof that well-typed terms are reducible is by structural induction on the derivation of the typing judgment. We will proceed by first proving lemmas that show the typing rules preserve reducibility, concluding at the end with the fundamental theorem.
Once again, we will focus our attention on the results corresponding to collection types, as the rest are standard.

Reducibility of singletons is trivial by definition, while that of empty collections is proved
in the same style as~\cite{Cooper09a}, with the obvious adaptations.
\begin{lem}[reducibility for singletons]\label{lem:redsingle}
For all $\cC$, if $M \in \cC$, then $\setlit{M} \in \toptopset{\cC}$.
\end{lem}
\begin{proof}
  Trivial by definition of $\top$-lifting and $\top\top$-lifting.
\end{proof}
\begin{lem}\label{lem:Kredempty}
If $K \in \SN$ is a continuation, then for all indices $p$ we have $K[p \mapsto \emptyoset] \in \SN$.
\end{lem}
\begin{cor}[reducibility for $\emptyoset$]\label{cor:redempty}
For all $\cC$, $\emptyoset \in \toptopset{\cC}$.
\end{cor}

As for unions, we will prove a more general statement on auxiliary continuations.
\begin{lem}\label{lem:Qredunion}\label{lem:Kredunion}\mbox{}\\
For all auxiliary continuations $Q,O_1,O_2$ with pairwise disjoint supports, if $Q[p \mapsto O_1] \in \SN$ and $Q[p \mapsto O_2] \in \SN$, then $Q[p \mapsto O_1 \ocup O_2] \in \SN$.
\end{lem}
The proof of the lemma above follows the same style as~\cite{Cooper09a}; however since our definition of auxiliary continuations is more general, the theorem statement mentions $O_1, O_2$ rather than pure terms: the hypothesis on the supports of the continuations being disjoint is required by this generalization.

\begin{cor}[reducibility for unions]\label{cor:redunion}
If $M \in \toptopset{\cC}$ and $N \in \toptopset{\cC}$, then $M \ocup N \in \toptopset{\cC}$.
\end{cor}

In some of the following proofs, a result that mirrors Lemma~\ref{lem:Qredunion} will also be useful.
\begin{lem}\label{lem:Kredunionr}
If $Q[p \mapsto M \ocup N] \in \SN$, then $Q[p \mapsto M] \in \SN$ and $Q[p \mapsto N] \in \SN$; furthermore, we have:
\begin{gather*}
\maxred{Q[p \mapsto M]} \leq \maxred{Q[p \mapsto M \ocup N]}
\\
\maxred{Q[p \mapsto N]} \leq \maxred{Q[p \mapsto M \ocup N]}
\end{gather*}
\end{lem}
\begin{proof}
We assume $p \in \supp(Q)$ (otherwise, $Q[p \mapsto M] = Q[p \mapsto N] = Q[p \mapsto M \ocup N]$, and the thesis holds trivially), then we show that any contraction in $Q[p \mapsto M]$ has a corresponding non-empty reduction sequence in $Q[p \mapsto M \ocup N]$, and the two reductions preserve the term form, therefore no reduction sequence of $Q[p \mapsto M]$ is longer than the maximal one in $Q[p \mapsto M \ocup N]$. The same reasoning applies to $Q[p \mapsto N]$.
\end{proof}

Like in proofs based on standard $\top\top$-lifting, the most challenging cases are those dealing with commuting conversions -- in our case, comprehensions and conditionals.
\begin{lem}\label{lem:redcomprmain}
Let $K$, $\pure{L}$, $\pure{N}$ be such that $K[p \mapsto \pure{N}\subst{x}{\pure{L}}] \in \SN$ and $\pure{L} \in \SN$. Then $K[p \mapsto \comprehension{\pure{N}|x \gets \setlit{\pure{L}}}] \in \SN$.
\end{lem}
\begin{proof}
In this proof, we assume the names of bound variables are chosen so as to avoid duplicates, and are distinct from the free variables. We proceed by well-founded induction on $(K,p,\pure{N},\pure{L})$ using the following metric:
\[
\begin{array}{l}
(K_1,p_1,\pure{N_1},\pure{L_1}) \prec (K_2,p_2,\pure{N_2},\pure{L_2})
\\
\iff
(\maxred{K_1[p_1 \mapsto \pure{N_1}\subst{x}{\pure{L_1}}]} + \maxred{\pure{L_1}}, \sz{K_1}_{p_1},\kwsize(\pure{N_1}))
\\
\qquad\quad \lessdot
(\maxred{K_2[p_2 \mapsto \pure{N_2}\subst{x}{\pure{L_2}}]} + \maxred{\pure{L_2}}, \sz{K_2}_{p_2},\kwsize(\pure{N_2}))
\end{array}
\]
Now we show that every contractum must be a strongly normalizing:
\begin{itemize}
\item $K[p \mapsto \pure{N}\subst{x}{\pure{L}}]$: this term is s.n.\ by hypothesis.
\item $K'[p \mapsto \comprehension{N|x \gets \setlit{\pure{L}}}]^\sigma$, where $K \ared{\sigma} K'$. Lemma~\ref{lem:redmaxred} allows us to prove
$\maxred{K'[p \mapsto \pure{N}\subst{x}{\pure{L}}]^\sigma} < \maxred{K[p \mapsto \pure{N} \subst{x}{\pure{L}}]}$ (since the former is a contractum of the latter),
which implies $\maxred{K'[q \mapsto \pure{N}\subst{x}{\pure{L}}]} \leq \maxred{K'[p \mapsto \pure{N}\subst{x}{\pure{L}}]^\sigma} < \maxred{K[p \mapsto \pure{N} \subst{x}{\pure{L}}]}$ for all $q$ s.t.\ $\sigma(q) = p$
by means of Lemma~\ref{lem:maxredctxapp} (because $[q \mapsto \pure{N}\subst{x}{\pure{L}}]$ is a subinstantiation of $[p \mapsto \pure{N} \subst{x}{\pure{L}}]^\sigma$);
then we can apply the IH to obtain, for all $q$ s.t.\ $\sigma(q) = p$, $K'[q \mapsto \comprehension{\pure{N} | x \gets \setlit{\pure{L}}}] \in \SN$;
by Corollary~\ref{cor:aredforall}, this implies the thesis.
\item $K[p \mapsto \emptyoset]$ (when $N = \emptyoset$): this is equal to $K[p \mapsto \emptyoset \subst{x}{\pure{L}}]$, which is s.n.\ by hypothesis.
\item $K[p \mapsto \comprehension{\pure{N_1}|x \gets \setlit{\pure{L}}} \ocup \comprehension{\pure{N_2}|x \gets \setlit{\pure{L}}}]$ (when $\pure{N} = \pure{N_1} \ocup \pure{N_2}$);
by IH (since $\kwsize(\pure{N_i}) < \kwsize(\pure{N_1} \ocup \pure{N_2})$, and all other metrics do not increase) we prove
$K[p \mapsto \comprehension{\pure{N_i}|x \gets \setlit{\pure{L}}}] \in \SN$ (for $i=1,2$),
and consequently obtain the thesis by Lemma~\ref{lem:Kredunion}.
\item $K_0[p \mapsto \comprehension{\comprehension{\pure{M}| y \gets \pure{N}}| x \gets \setlit{\pure{L}}}]$, where $K = K_0 \compop{p} \comprehension{\pure{M}|y \gets \Box}$;
since we know, by the hypothesis on the choice of bound variables, that $x \notin \FV(\pure{M})$,
we note that $K_0[p \mapsto \comprehension{\pure{M}|y \gets \pure{N}}\subst{x}{\pure{L}}] = K[p \mapsto \pure{N}\subst{x}{\pure{L}}]$;
furthermore, by Lemma~\ref{lem:szdecrease} we know $\sz{K_0}_p < \sz{K}_p$; then we can apply the IH to obtain the thesis.
\item $K_0[p \mapsto \comprehension{\plwhere~\pure{B}~\kwdo~\pure{N} | x \gets \setlit{\pure{L}}}]$ (when $K = K_0 \compop{p} \plwhere~\pure{B}~\kwdo~\Box$):
since we know, from the hypothesis on the choice of bound variables, that $x \notin \FV(B)$,
we note that $K_0[p \mapsto (\plwhere~\pure{B}~\kwdo~\pure{N})\subst{x}{\pure{L}}] = K[p \mapsto \pure{N}\subst{x}{\pure{L}}]$;
furthermore, by Lemma~\ref{lem:szdecrease} we know $\sz{K_0}_p < \sz{K}_p$; then we can apply the IH to obtain the thesis.
\item reductions within $N$ or $L$ follow from the IH by reducing the induction metric.
\qedhere
\end{itemize}
\end{proof}

\begin{lem}[reducibility for comprehensions]\label{lem:redcomprehension}
Assume $\CRi(\cC)$, $\CRi(\cD)$, $\pure{M} \in \toptopset{\cC}$ and for all $\pure{L} \in \cC$, $\pure{N}\subst{x}{\pure{L}} \in \toptopset{\cD}$. Then $\comprehension{\pure{N} | x \gets \pure{M}} \in \toptopset{\cD}$.
\end{lem}
\begin{proof}
We assume $p$, $K \in \topset{\cD_p}$ and prove $K[p \mapsto \comprehension{\pure{N} | x \gets \pure{M}}] \in \SN$. We start by showing that $K' = K \compop{p} \comprehension{\pure{N}|x \gets \Box} \in \topset{\cC_p}$, or equivalently that for all $\pure{L} \in \cC$, $K'[p \mapsto \setlit{\pure{L}}] = K[p \mapsto \comprehension{\pure{N}|x \gets \setlit{\pure{L}}}] \in \SN$: since $\CRi(\cC)$, we know $\pure{L} \in \SN$, and since $\pure{N}\subst{x}{\pure{L}} \in \toptopset{\cD}$, $K[p \mapsto \pure{N}\subst{x}{\pure{L}}] \in \SN$; then we can apply Lemma~\ref{lem:redcomprmain} to obtain $K'[p \mapsto \setlit{\pure{L}}] \in \SN$ and consequently $K' \in \topset{\cC_p}$. But then, since $\pure{M} \in \toptopset{\cC}$, we have $K'[p \mapsto \pure{M}] = K[p \mapsto \comprehension{\pure{N} | x \gets \pure{M}}] \in \SN$, which is what we needed to prove.
\end{proof}

Reducibility for conditionals is proved in a similar manner. However, to make the induction work under all the conversions commuting with $\plwhere$, we cannot prove the strong normalization statement within regular continuations $K$, but we need to generalize it to auxiliary continuations. A minor complication with the merging of nested $\plwhere$ is handled by a separate lemma. Additionally, due to the more complicated structure of auxiliary continuations, we will need to ensure that the free variables of the Boolean guard of the $\plwhere$ expression do not get captured: the assumption uses an auxiliary operation $\BV$ denoting the set of variables bound over holes:
\begin{defi}\label{def:bv}
The operation $\BV(Q)$ is defined as follows:
\begin{align*}
\BV(\ibox{p}) & = \BV(\pure{M}) = \emptyset
\\
\BV(Q_1 \cup Q_2) & = \BV(Q_1) \cup \BV(Q_2)
\\
\BV(\plwhere~\pure{B}~\kwdo~Q) & = \BV(Q)
\\
\BV(\comprehension{Q_1 \mid x \gets Q_2}) & = \left\{ \begin{array}{ll}
	\metaset{x} \cup \BV(Q_1) \cup \BV(Q_2) & \text{if $\supp(Q_1) \neq \emptyset$}
	\\
	\BV(Q_1) \cup \BV(Q_2) & \text{otherwise}
	\end{array} \right.
\end{align*}
\end{defi}

\begin{lem}\label{lem:redwherechange}
  Suppose $Q[p \mapsto \plwhere~B~\kwdo~M] \in \SN$. Then for all $B' \in \SN$ such that $\BV(Q)$ and $\FV(B')$ are disjoint, $Q[p \mapsto \plwhere~B \land B'~\kwdo~M] \in \SN$.
\end{lem}
\begin{lem}\label{lem:redwheremain}
Let $Q$, $B$, $O$ such that $Q[p \mapsto O] \in \SN$, $B \in \SN$, $\BV(Q) \cap \FV(B) = \emptyset$ and $\supp(Q) \cap \supp(O) = \emptyset$. Then $Q[p \mapsto \plwhere~B~\kwdo~O] \in \SN$.
\end{lem}
\begin{proof}
  In this proof, we assume the names of bound variables are chosen so as to avoid duplicates, and distinct from the free variables. It is important to notice that this is the main proof in which auxiliary continuations, as opposed to regular continuations, are needed to obtain a usable induction hypothesis when the argument of $\plwhere$ happens to be a comprehension.
  We proceed by well-founded induction on $(Q, B, O, p)$ using the following metric:
  \[
  \begin{array}{l}
  (Q_1,B_1,O_1,p_1) \prec (Q_2,B_2,O_2,p_2)
  \iff
  \\
  \qquad
  (\maxred{Q_1[p_1 \mapsto O_1]}, \len{Q_1}_{p_1}, \kwsize(O_1), \maxred{B_1})
  \\
  \qquad
  \lessdot
  (\maxred{Q_2[p_2 \mapsto O_2]}, \len{Q_2}_{p_2}, \kwsize(O_2), \maxred{B_2})
  \end{array}
  \]
  We will consider all possible contracta and show that each of them must be a strongly normalizing term; we will apply the induction hypothesis to new auxiliary continuations obtained by placing pieces of $O$ into $Q$ or vice-versa: the hypothesis on the supports of $Q$ and $O$ being disjoint is used to make sure that the new continuations do not contain duplicate holes and are thus well-formed. By cases on the possible contracta:
  \begin{itemize}
  \item $Q_1[q \mapsto Q_2 \subst{x}{\pure{L}}][p \mapsto (\plwhere~B~\kwdo~O)\subst{x}{\pure{L}}]$, where $Q = (Q_1 \compop{q} \comprehension{\Box \mid x \gets \setlit{\pure{L}}})[q \mapsto Q_2]$, $q \in \supp(Q_1)$, and $p \in \supp(Q_2)$; by the freshness condition we know $x \notin \FV(B)$, thus $(\plwhere~B~\kwdo~O)\subst{x}{\pure{L}} = \plwhere~B~\kwdo~(O\subst{x}{\pure{L}})$; we take $Q' = Q_1[q \mapsto Q_2\subst{x}{\pure{L}}]$ and $O' = O\subst{x}{\pure{L}}$, and note that $\maxred{Q'[p \mapsto O']} < \maxred{Q[p \mapsto O]}$, because the former term is a contractum of the latter: then we can apply the IH to prove $Q'[p \mapsto \plwhere~B~\kwdo~O'] \in \SN$, as needed.
  \item $Q'[p \mapsto \plwhere~B~\kwdo~O]^\sigma$, where $Q \ared{\sigma} Q'$. We know $\maxred{Q'[p \mapsto O]^\sigma} < \maxred{Q[p \mapsto O]}$ by Lemma~\ref{lem:redmaxred} since the latter is a contractum of the former. By Corollary~\ref{cor:aredforall}, for all $q$ s.t.\ $\sigma(q) = p$ we have $\maxred{Q'[q \mapsto O]} \leq \maxred{Q'[p \mapsto O]^\sigma}$; we can thus apply the IH to obtain $Q[q \mapsto \plwhere~B~\kwdo~O] \in \SN$ whenever $\sigma(q) = p$. By Corollary~\ref{cor:aredforall}, this implies the thesis.
  \item $Q_1[p \mapsto \plwhere~B~\kwdo~\comprehension{Q_2|x \gets O}]$, where $Q = Q_1 \compop{p} \comprehension{Q_2|x \gets \Box}$; we take
  \linebreak
  $O' = \comprehension{Q_2|x \gets O}$, and we note that $Q[p \mapsto O] = Q_1[p \mapsto O']$ and, by Lemma~\ref{lem:szdecrease}, $\len{Q_1}_p < \len{Q}_p$; we can thus apply the IH to prove $Q_1[p \mapsto \plwhere~B~\kwdo~O'] \in \SN$, as needed.
  \item $Q_0[p \mapsto \plwhere~(B_0 \land B)~\kwdo~O]$, where $Q = Q_0 \compop{p} (\plwhere~B_0~\kwdo~\Box)$; we know by hypothesis that $Q_0[p \mapsto \plwhere~B_0~\kwdo~O] \in \SN$ and $B \in \SN$; then the thesis follows by Lemma~\ref{lem:redwherechange}.
  \item $Q[p \mapsto \emptyoset]$, where $O = \emptyoset$: this term is strongly normalizing by hypothesis.
  \item $Q[p \mapsto (\plwhere~B~\kwdo~O_1) \ocup (\plwhere~B~\kwdo~O_2)]$, where $O = O_1 \ocup O_2$; for $i=1,2$, we prove $Q[p \mapsto O_i] \in \SN$ and $\maxred{Q[p \mapsto O_i]} \leq \maxred{Q[p \mapsto O]}$ by Lemma~\ref{lem:Kredunionr}, and we also note $\kwsize(O_i) < \kwsize(O)$; then we can apply the IH to prove $Q[p \mapsto \plwhere~B~\kwdo~O_i] \in \SN$, which implies the thesis by Lemma~\ref{lem:Kredunion}.
  \item $Q[p \mapsto \comprehension{\plwhere~B~\kwdo~O_1|x \gets O_2}]$, where $O = \comprehension{O_1|x \gets O_2}$;
  we take
  \[Q' = Q \compop{p} \comprehension{\Box \mid x \gets O_2}\] and we have that $Q'[p \mapsto \plwhere~B~\kwdo~O_1]$ is equal to $Q[p \mapsto \comprehension{\plwhere~B~\kwdo~O_1|x \gets O_2}]$;
  %we thus note $\maxred{Q'[p \mapsto O_1]} = \maxred{Q[p \mapsto \comprehension{O_1|x \gets O_2}]} = \maxred{Q[p \mapsto O]}$,
  we thus note $\maxred{Q'[p \mapsto O_1]} = \maxred{Q[p \mapsto O]}$,
  $\len{Q'}_p = \len{Q}_p$ (Lemma~\ref{lem:szdecrease}), and $\kwsize(O_1) < \kwsize(O)$, thus we can apply the IH to prove $Q'[p \mapsto \plwhere~B~\kwdo~O_1] \in \SN$, as needed. We remark that in this subcase it is essential that the IH be generalized to auxiliary continuations, because even if we assume that $Q$ is a regular continuation $K$ and $O_2$ is a pure term $\pure{L}$, $K \compop{p} \comprehension{\Box \mid x \gets \pure{L}}$ is \emph{not} a regular continuation.
  \item $Q[p \mapsto \plwhere~(B \land B_0)~\kwdo~O_0]$, where $O = \plwhere~B_0~\kwdo~O_0$; we know by hypothesis that $Q[p \mapsto \plwhere~B_0~\kwdo~O_0] \in \SN$ and $B \in \SN$; then the thesis follows by Lemma~\ref{lem:redwherechange}.
  \item Reductions within $B$ or $O$ make the induction metric smaller, thus follow immediately from the IH\@.
  \qedhere
  \end{itemize}
  \end{proof}

\begin{lem}\label{lem:bvK}
For all regular continuations $K$, $\BV(K) = \emptyset$.
\end{lem}
\begin{proof}
This follows immediately by noticing that in regular continuations $K$ (unlike auxiliary continuations $Q$) holes never appear in the head of a comprehension.
\end{proof}

\begin{cor}[reducibility for conditionals]\label{cor:redwhere}\mbox{}\\
If $\pure{B} \in \SN$ and $\pure{N} \in \Red{\setlit{T}}$, then $\plwhere~\pure{B}~\kwdo~\pure{N} \in \Red{\setlit{T}}$.
\end{cor}
\begin{proof}
We need to prove that for all $K \in \topset{\Red{T}}$ we have $K[\plwhere~\pure{B}~\kwdo~\pure{N}] \in \SN$. By Lemma~\ref{lem:bvK}, we prove $\BV(K) = \emptyset$; then we apply Lemma~\ref{lem:redwheremain} with $Q = K$ to obtain the thesis.
\end{proof}

Finally, reducibility for the emptiness test is proved in the same style as~\cite{Cooper09a}.

\begin{lem}\label{lem:redisempty}
For all $M$ and $T$ such that $\Gamma \vdash M : \setlit{T}$ and $M \in \toptopset{\Red{T}}$, we have $\isempty(M) \in \SN$.
\end{lem}

\subsection{Main theorem}
Before stating and proving the main theorem, we introduce some auxiliary notation.

\begin{defi}\mbox{}
\begin{enumerate}
\item A substitution $\rho$ satisfies $\Gamma$ (notation:
$\rho \vDash \Gamma$) iff, for all $x \in \dom(\Gamma)$, $\rho(x) \in
\Red{\Gamma(x)}$.
\item A substitution $\rho$ satisfies $M$ with type $T$ (notation: $\rho \vDash M : T$) iff $M \rho \in
\Red{T}$.
\end{enumerate}
\end{defi}

\noindent
As usual, the main result is obtained as a corollary of a stronger theorem generalized to substitutions into open terms, by using the identity substitution $\id_\Gamma$.
\begin{lem}\label{lem:idsubst}
For all $\Gamma$, we have $\id_\Gamma \vDash \Gamma$.
\end{lem}
\begin{thm}\label{thm:reducibility}
If $\Gamma \vdash M : T$, then for all $\rho$ such that $\rho \vDash \Gamma$, we have $\rho \vDash M : T$
\end{thm}
\begin{proof}
By induction on the derivation of $\Gamma \vdash M : T$. When $M$ is a singleton, an empty collection, a union, a conditional, or an emptiness test, we use Lemmas~\ref{lem:redsingle} and~\ref{lem:redisempty}, and Corollaries~\ref{cor:redempty},~\ref{cor:redunion}, and~\ref{cor:redwhere}. For comprehensions such that $\Gamma \vdash \comprehension{M_1 | x \gets M_2} : \setlit{T}$, we know by IH that $\rho \vDash M_2 : \setlit{S}$ and for all $\rho' \vDash \Gamma, x:S$ we have $\rho' \vDash M_1 : \setlit{T}$: we prove that for all $L \in \Red{S}$, $\rho\subst{x}{L} \vDash \Gamma, x:S$, hence $\rho\subst{x}{L} \vDash M_1 : \setlit{T}$; then we obtain $\rho \vDash \comprehension{M_1 | x \gets M_2} : \setlit{T}$ by Lemma~\ref{lem:redcomprehension}. Non-collection cases are standard.
\end{proof}
\begin{cor}
If $\Gamma \vdash M : T$, then $M \in \SN$.
\end{cor}

\section{Heterogeneous Collections}\label{sec:heterogeneous}
SQL allows a user to write queries that will evaluate to relations that are bags of tuples by means of constructs including $\kwsel$ statements and $\kwunion~\kwall$ operations; additionally, it also allows constructs like $\kwsel~\kwdist$ and $\kwunion$ to produce sets of tuples (more precisely, bags without duplicates); both kinds of constructs can be freely mixed in the same query. In contrast, the language \NRCl we have discussed in the previous sections can only deal with one kind of collection (either sets or bags).

In a short paper~\cite{ricciotti19dbpl}, we introduced a generalization of \NRC called
\NRCsb that makes up for this shortcoming by allowing both set-valued and bag-valued
collections (with distinct types denoted by $\setlit{T}$ and $\msetlit{T}$),
along with mappings from bags to sets (deduplication $\distinct$) and
from sets to bags (promotion $\promote$).  We conjectured that this
language also satisfies a normalization property, allowing its normal forms to be translated to SQL\@. Here, we prove that \NRCsb is, indeed, strongly normalizing,
even when extended to a richer language \NRClsb with
higher-order (nonrecursive) functions. Its syntax is a straightforward extension of \NRCl:
\[
\begin{array}{lrcl}
  \mathbf{types} & S, T & ::= & \ldots \orelse \msetlit{T} \\
  \mathbf{terms} & L, M, N & ::= & \ldots \orelse \emptymset \orelse \msetlit{M} \orelse M \uplus N \orelse \biguplus \msetlit{M | x \leftarrow N} \\
      &  & \orelse & \bagwhere~M~\kwdo~N \orelse \distinct M \orelse \promote M
\end{array}
\]

We use $\msetlit{T}$ to denote the type of bags containing elements of type $T$; similarly, the notations $\emptymset$, $\msetlit{M}$, $M \uplus N$,
$\biguplus\msetlit{M | x \gets N}$ denote
%respectively
empty and singleton bags,
bag disjoint union and bag comprehension; the language also includes conditionals on bags. The notations $\promote M$ and $\distinct N$ stand, respectively, for the bag containing exactly one copy of each element of the set $M$, and for the set containing the elements of the bag $N$, forgetting about their multiplicity. We do not need to provide a primitive emptiness test for bags, since it can be defined anyway as $\bagempty~M := \plempty~\distinct M$.

\begin{figure}[tb]
  \begin{center}
  \AxiomC{$\phantom{A}$}
  \UnaryInfC{$\Gamma \vdash \emptymset : \msetlit{T}$}
  \DisplayProof
  \hspace{.5cm}
  \AxiomC{$\Gamma \vdash M : T$}
  \UnaryInfC{$\Gamma \vdash \msetlit{M} : \msetlit{T}$}
  \DisplayProof
  \hspace{.5cm}
  \AxiomC{$\Gamma \vdash M : \msetlit{T}$}
  \AxiomC{$\Gamma \vdash N : \msetlit{T}$}
  \BinaryInfC{$\Gamma \vdash M \uplus N : \msetlit{T}$}
  \DisplayProof

  \medskip

  \AxiomC{$\Gamma, x:T \vdash M : \msetlit{S}$}
  \AxiomC{$\Gamma \vdash N : \msetlit{T}$}
  \BinaryInfC{$\Gamma \vdash \mcomprehension{M | x \leftarrow N} : \msetlit{S}$}
  \DisplayProof
  \hspace{.5cm}
  \AxiomC{$\Gamma \vdash M : \boolty$}
  \AxiomC{$\Gamma \vdash N : \msetlit{T}$}
  \BinaryInfC{$\Gamma \vdash \bagwhere~M~\kwdo~N : \msetlit{T}$}
  \DisplayProof

  \medskip

  \AxiomC{$\Gamma \vdash M : \msetlit{T}$}
  \UnaryInfC{$\Gamma \vdash \distinct M : \setlit{T}$}
  \DisplayProof
  \hspace{.5cm}
  \AxiomC{$\Gamma \vdash M : \setlit{T}$}
  \UnaryInfC{$\Gamma \vdash \promote M : \msetlit{T}$}
  \DisplayProof
  \end{center}
  \caption{Additional typing rules for \NRClsb.}\label{fig:NRClsb_typing}
  \end{figure}

The type system for \NRClsb is obtained from the one for \NRCl by adding the unsurprising rules of Figure~\ref{fig:NRClsb_typing}: these largely replicate, at the bag level, the corresponding set-based rules; additionally, the rules for $\distinct$ and $\promote$ describe how these operators turn bag-typed terms into set-typed ones, and vice-versa. Similarly, the rewrite system for \NRClsb is also an extension of the one for \NRCl, with additional reduction rules for the new operators involving bags that mimic the corresponding set-based operations; there are simplification rules involving $\distinct$ that state that the deduplication of empty or singleton bags yields empty or singleton sets, and that deduplication commutes with bag union and comprehension, turning them into their set counterparts. The promotion of empty or singleton sets can be simplified away in a symmetric way; however, promotion does \emph{not} commute with union and comprehension (this avoids contractions like $\promote(\setlit{x} \cup \setlit{x}) \not\red \promote\setlit{x} \uplus \promote\setlit{x}$, which would be unsound in the intended model, where $\cup$ is idempotent, but $\uplus$ is not). These reduction rules are described in Figure~\ref{fig:norm_NRClsb}.

\begin{figure}[tb]
    \begin{center}\small{
    \begin{tabular}{r@{$~\red~$}ll}
    \multicolumn{3}{c}{$
    \mcomprehension{\emptymset | x \leftarrow M} \red \emptymset
    \qquad
    \mcomprehension{M | x \leftarrow \emptymset} \red \emptymset
    \qquad
    \mcomprehension{M | x \leftarrow \msetlit{N}} \red M[N/x]
    $}
    \\
    $\mcomprehension{M \uplus N | x \leftarrow R}$
     & $\mcomprehension{M | x \leftarrow R} \uplus \mcomprehension{N | x \leftarrow R}$ &
    \\
    $\mcomprehension{M | x \leftarrow N \uplus R}$
    & $\mcomprehension{M | x \leftarrow N} \uplus \mcomprehension{M | x \leftarrow R}$ &
    \\
    $\mcomprehension{M | y \leftarrow \mcomprehension{R | x \leftarrow N }}$
    & $\mcomprehension{ M | x \leftarrow N, y \leftarrow R}$ & (if $x \notin \FV(M)$)
    \\
    $\mcomprehension{M | x \leftarrow \bagwhere~N~\kwdo~R}$
    & $\mcomprehension{\bagwhere~N~\kwdo~M | x \leftarrow R}$ & (if $x \notin \FV(M)$)
    \\
    \multicolumn{3}{c}{} %% BLANK
    \\
    \multicolumn{3}{c}{$
    \bagwhere~\kwtrue~\kwdo~M \red M
    \qquad
    \bagwhere~\kwfalse~\kwdo~M \red \emptymset
    \qquad
    \bagwhere~M~\kwdo~\emptymset \red \emptymset
    $}
    \\
    $\bagwhere~M~\kwdo~(N \uplus R)$ & $(\bagwhere~M~\kwdo~N) \uplus (\bagwhere~M~\kwdo~R)$ &
    \\
    $\bagwhere~M~\kwdo~\mcomprehension{N | x \leftarrow R}$
    & $\mcomprehension{\bagwhere~M~\kwdo~N | x \leftarrow R}$ &
    \\
    $\bagwhere~M~\kwdo~\bagwhere~N~\kwdo~R$ & $\bagwhere~(M \land N)~\kwdo~R$ &
    \end{tabular}
    }
    \end{center}
    \[
\begin{array}{c}
\distinct\emptymset \red \emptyset \qquad
\distinct\msetlit{M} \red \setlit{M} \qquad
\distinct(M \uplus N) \red \distinct M \cup \distinct N \qquad
\distinct\promote M \red M
\\
\distinct\biguplus\msetlit{M | x \gets N} \red
  \comprehension{\distinct M | x \gets \distinct N} \qquad
\distinct(\bagwhere~M~\kwdo~N) \red \plwhere~M~\kwdo~\distinct N
\\
\promote\emptyset \red \emptymset \qquad
\promote\setlit{M} \red \msetlit{M} \qquad
\promote(\plwhere~M~\kwdo~N) \red \bagwhere~M~\kwdo~\promote N
\end{array}
\]
\caption{Additional rewrite rules for \NRClsb.}\label{fig:norm_NRClsb}
\end{figure}

An obvious characteristic of \NRClsb, compared to \NRCl, is the duplication of syntax caused by the presence of two separate types of collections. A direct proof of strong normalization of this calculus would require us to consider many more cases than we have seen in \NRCl. A more efficient approach is to show that the strong normalization property of \NRClsb descends, as a corollary, from the strong normalization of a slightly tweaked version of \NRCl, comprising a single type of collections, but also retaining the $\distinct$ and $\promote$ operators of \NRClsb. This is the formalism \NRCldi described in the next subsection.

\subsection{The simplified language \texorpdfstring{\NRCldi}{NRC-λδι}}
The simplified language \NRCldi is obtained from \NRCl by adding the two operators $\distinct$, $\promote$, and nothing else:
\[
\begin{array}{rcl}
  L, M, N & ::= & \ldots \orelse  \distinct M \orelse \promote M
\end{array}
\]
\NRCldi does not add any type compared to \NRCl: in particular, if $M$ has type $\setlit{T}$, then $\distinct M$ and $\promote M$ have type $\setlit{T}$ as well. The rewrite system extends \NRCl with straightforward adaptations of the \NRClsb rules involving $\distinct$ and $\promote$. All of the additional typing and rewrite rules are shown in Figure~\ref{fig:rules_NRCldi}.

\begin{figure}[tb]
  \begin{center}
  \AxiomC{$\Gamma \vdash M : \setlit{T}$}
  \UnaryInfC{$\Gamma \vdash \distinct M : \setlit{T}$}
  \DisplayProof
  \hspace{.5cm}
  \AxiomC{$\Gamma \vdash M : \setlit{T}$}
  \UnaryInfC{$\Gamma \vdash \promote M : \setlit{T}$}
  \DisplayProof
  \end{center}

  \bigskip

\[
\begin{array}{c}
\distinct\emptyset \red \emptyset \qquad
\distinct\setlit{M} \red \setlit{M} \qquad
\distinct(M \cup N) \red \distinct M \cup \distinct N \qquad
\distinct\promote M \red M
\\
\distinct\comprehension{M | x \gets N} \red
  \comprehension{\distinct M | x \gets \distinct N} \qquad
\distinct(\plwhere~M~\kwdo~N) \red \plwhere~M~\kwdo~\distinct N
\\
\promote\emptyset \red \emptyset \qquad
\promote\setlit{M} \red \setlit{M} \qquad
\promote(\plwhere~M~\kwdo~N) \red \plwhere~M~\kwdo~\promote N
\end{array}
\]
\caption{Additional typing and rewrite rules for \NRCldi.}\label{fig:rules_NRCldi}
\end{figure}

\begin{figure}[tb]
\begin{gather*}
  \erase{A} = A \qquad \erase{S \imp T} = \erase{S} \imp \erase{T} \qquad
  \erase{\tuple{\vect{\ell : T}}} = \tuple{\vect{\ell : \erase{T}}} \qquad
  \erase{\setlit{T}} = \erase{\msetlit{T}} = \setlit{\erase{T}}
  \\
  \erase{x_1 : T_1,\ldots, x_n : T_n} = x_1 : \erase{T_1}, \ldots, x_n : \erase{T_n}
\end{gather*}
\begin{align*}
  \erase{x} & = x
  &
  \erase{c(\vect{M})} & = c(\vect{\erase{M}})
  \\
  \erase{\tuple{\vect{\ell = M}}} & = \tuple{\vect{\ell = \erase{M}}}
  &
  \erase{M.\ell} & = \erase{M}.\ell
  \\
  \erase{\lambda x.M} & = \lambda x.\erase{M}
  &
  \erase{(M~N)} & = (\erase{M}~\erase{N})
  \\
  \erase{\emptyset} & = \erase{\emptymset} = \emptyset
  &
  \erase{\setlit{M}} & = \erase{\msetlit{M}} = \setlit{\erase{M}}
  \\
  \erase{M \cup N} & = \erase{M \uplus N} = \erase{M} \cup \erase{M}
  &
  \erase{\plempty~M} & = \plempty~\erase{M}
\end{align*}
\begin{align*}
  \erase{\comprehension{M \mid x \gets N}} & = \erase{\mcomprehension{M \mid x \gets N}} = \comprehension{\erase{M} \mid x \gets \erase{N}}
  \\
  \erase{\plwhere~M~\kwdo~N} & = \erase{\bagwhere~M~\kwdo~N} = \plwhere~\erase{M}\kwdo~\erase{N}
\end{align*}
\caption{Forgetful translation of \NRClsb into \NRCldi.}\label{fig:xlatedef}
\end{figure}

\NRClsb types and terms can be translated to \NRCldi by means of a forgetful operation $\erase{\cdot}$, described in Figure~\ref{fig:xlatedef}.
A straightforward induction is sufficient to prove that this translation preserves typability and reduction.

\begin{thm}\label{thm:xlatetyp}
If $\Gamma \vdash M : T$ in \NRClsb, then $\erase{\Gamma} \vdash \erase{M} : \erase{T}$
in \NRCldi.
\end{thm}
\begin{thm}\label{thm:xlatered}
For all terms $M$ of \NRClsb, if $M \red M'$, we have $\erase{M} \red \erase{M'}$
in \NRCldi. Consequently, if $\erase{M} \in \SN$ in \NRCldi, then $M \in \SN$ in \NRClsb.
\end{thm}

Thanks to the two results above, strong normalization for \NRClsb is an immediate consequence of strong normalization for \NRCldi.

\subsection{Reducibility for \texorpdfstring{\NRCldi}{NRC-λδι}}
We are now going to present an extension of the strong normalization proof for \NRCl, allowing us to derive the same result for \NRCldi (and, consequently, for \NRClsb). Concretely, this extension involves adding some extra cases to some definitions and proofs; in a single case, we need to strengthen the statement of a lemma, whose proof remains otherwise close to its \NRCl version.

For \NRCldi continuations and frames, we will allow extra cases including $\distinct$ and $\promote$, as follows:
\begin{align*}
K,H ::= & \ldots \orelse \distinct K \orelse \promote K
\\
Q,O ::= & \ldots \orelse \distinct Q \orelse \promote Q
\\
F ::= & \ldots \orelse \distinct \Box \orelse \promote \Box
\end{align*}
We also extend the measures $\len{Q}_p$ and $\sz{Q}_p$ to account for the new cases.
\begin{defi}[extends~\ref{def:measures}]\label{def:hmeasures}
  The measures $\len{Q}_p$ and $\sz{Q}_p$ of \NRCl are extended to \NRCldi by means of the following additional cases:
  \[
    \len{\distinct Q}_p = \len{\promote Q}_p = \len{Q}_p + 1
    \qquad
    \sz{\distinct Q}_p = \sz{\promote Q}_p = \sz{Q}_p + 1
  \]
\end{defi}
Renaming reduction in \NRCldi is defined in the same way as its \NRCl counterpart.

\medskip

We notice that terms in the form $\distinct M$, when plugged into a context, never create new redexes: we thus extend the definition of neutral terms.
\begin{defi}[extends~\ref{def:neutral}]\label{def:hneutral}
The grammar of the neutral terms of \NRCl is extended to \NRCldi by means of the following additional production:
\[
W ::= \ldots \orelse \distinct M
\]
\end{defi}

Since the type sublanguage of \NRCldi is the same as in \NRCl, we can \emph{superficially} reuse the definition of reducibility sets: however, it is intended that the terms and continuations appearing in these definitions are those of \NRCldi rather than \NRCl. Similarly, the various technical lemmas involving contexts, continuations and instantiations use a uniform proof style that works seamlessly in \NRCldi; however, it is worth mentioning that Lemma~\ref{lem:szdecrease} holds because the definitions of the measures $\len{\cdot}$, $\sz{\cdot}$, and that of frames are aligned in \NRCldi just like in \NRCl; and that the proof of Lemma~\ref{lem:classification} must accommodate the additional frames of \NRCldi in the reduction at the interface. The proofs showing that all the reducibility sets are candidates (Lemmas~\ref{lem:cr1set} (CR1),~\ref{lem:cr2set} (CR2), and~\ref{lem:cr3set} (CR3)), use \NRCldi terms and continuations, but do not need to change structurally (Lemmas~\ref{lem:cr1set} and~\ref{lem:cr2set} do not need to inspect the shape of continuations and terms, while in Lemma~\ref{lem:cr3set} we do not need to consider any of the additional cases for an \NRCldi continuation $K$, because $K$ is applied to a neutral term, therefore there are no redexes at the interface regardless of the shape of $K$).

However, we do need to prove that the additional typing rules of \NRCldi (i.e.\ the introduction rules for $\distinct$ and $\promote$) preserve reducibility. This is expressed by the following results:
\begin{lem}\label{lem:Kreddistinct}
For all indices $p$ and candidates $\cC \in \CR$, if $K \in \topset{\cC_p}$, then $K \compop{p} (\distinct \Box) \in \topset{\cC_p}$.
\end{lem}
\begin{proof}
By unfolding the definitions, we prove that for all $p$, if $K \in \topset{\cC_p}$ and $M \in \cC$, then $K[p \mapsto \distinct\setlit{M}] \in \SN$.
We proceed by well-founded induction on $(K,M)$ using the following metric:
\[
(K_1,M_1) \prec (K_2,M_2) \iff (\maxred{K_1}, \maxred{M_1}) \lessdot (\maxred{K_2}, \maxred{M_2})
\]
Equivalently, we prove that all the contracta of $K[p \mapsto \distinct \setlit{M}]$ are s.n.:
\begin{itemize}
\item $K'[p \mapsto \distinct\setlit{M}]^\sigma$ (where $K \ared{\sigma} K'$): to prove this term is s.n., by Corollary~\ref{cor:aredforall} we need to show that $K'[p' \mapsto \distinct\setlit{M}] \in \SN$ for all $p'$ s.t.\ $\sigma(p') = p$; by Lemmas~\ref{lem:cr2setK} and~\ref{lem:Kcrsimpl}, we know $K' \in (\topset{\cC_{p'}})^\sigma$, and naturally $\maxred{K'} < \maxred{K}$ (Lemma~\ref{lem:redmaxred}), so the thesis follows by the IH\@.
\item $K[p \mapsto \distinct\setlit{M'}]$ (where $M \red M'$): by IH, with unchanged $K$, $M' \in \cC$ (Lemma~\ref{lem:cr2setK}), and $\maxred{M'} < \maxred{M}$ (Lemma~\ref{lem:redmaxred}).
\item $K[p \mapsto \setlit{M}]$: this is trivial by hypothesis.
\qedhere
\end{itemize}
\end{proof}
\begin{cor}[reducibility for $\distinct$]\label{cor:reddistinct}
For all $\cC \in \CR$, if $M \in \toptopset{\cC}$, then $\distinct M \in \toptopset{\cC}$.
\end{cor}
\begin{proof}
We need to prove that for all indices $p$, for all $K \in \topset{\cC}_p$, we have $K[p \mapsto \distinct M] \in \SN$. By Lemma~\ref{lem:Kreddistinct}, we prove $K \compop{p} (\distinct \Box) \in \topset{\cC}_p$; since $M \in \toptopset{\cC}$, we have $(K \compop{p} (\distinct \Box))[p \mapsto M] \in \SN$, which is equivalent to the thesis.
\end{proof}

\begin{lem}\label{lem:Kredpromote}
For all indices $p$ and candidates $\cC \in \CR$, if $K \in \topset{\cC_p}$, then $K \compop{p} (\promote \Box) \in \topset{\cC_p}$.
\end{lem}
\begin{proof}
  By unfolding the definitions, we prove that for all $p$, if $K \in \topset{\cC_p}$ and $M \in \cC$, then $K[p \mapsto \promote\setlit{M}] \in \SN$.
The proof follows the same steps as that of Lemma~\ref{lem:Kreddistinct}, but we have to consider an additional contractum for $K = K_0 \compop{p} (\distinct \Box)$:
\[
K[p \mapsto \promote \setlit{M}] = K_0[p \mapsto \distinct\promote\setlit{M}] \red K_0[p \mapsto \setlit{M}]
\]
Since $K \in \topset{\cC}_p$ and $M \in \cC$, we prove $K[p \mapsto \setlit{M}] = K_0[p \mapsto \distinct\setlit{M}]\in \SN$. Thus, $K_0[p \mapsto \setlit{M}] \in \SN$ as well, being a contractum of that term. This proves the thesis.
\end{proof}
\begin{cor}[reducibility for $\promote$]\label{cor:redpromote}
If $M \in \toptopset{\cC}$, then $\promote M \in \toptopset{\cC}$.
\end{cor}
\begin{proof}
We need to prove that for all indices $p$, for all $K \in \topset{\cC}_p$, we have $K[p \mapsto \promote M] \in \SN$. By Lemma~\ref{lem:Kredpromote}, we prove $K \compop{p} (\promote \Box) \in \topset{\cC}_p$; since $M \in \toptopset{\cC}$, we have $(K \compop{p} (\promote \Box))[p \mapsto M] \in \SN$, which is equivalent to the thesis.
\end{proof}

Finally we need to reconsider the reducibility properties of unions, comprehensions, and conditionals (Lemmas~\ref{lem:Qredunion},~\ref{lem:redcomprmain}, and~\ref{lem:redwheremain}), to add the extra cases in the updated definition of continuations. In the case of comprehensions, we need to reformulate the statement in a slightly strengthened way to ensure that the induction hypothesis remains applicable. The proofs concerning singletons (Lemma~\ref{lem:redsingle}) and empty sets (Corollary~\ref{cor:redempty}) do not need intervention.

\begin{lem}[extends~\ref{lem:Qredunion}]\label{lem:hQredunion}\mbox{}\\
  For all auxiliary continuations $Q,O_1,O_2$ with pairwise disjoint supports, if $Q[p \mapsto O_1] \in \SN$ and $Q[p \mapsto O_2] \in \SN$, then $Q[p \mapsto O_1 \ocup O_2] \in \SN$.
\end{lem}
\begin{proof}
For $Q = Q_0 \compop{p} (\distinct \Box)$, $Q[p \mapsto O_1 \ocup o_2]$ has an additional contractum $Q_0[p \mapsto \distinct O_1 \ocup \distinct O_2]$. We prove that $\maxred{Q_0} \leq \maxred{Q}$ and $\sz{Q_0}_p < \sz{Q}_p$: then we can use the IH to prove the thesis.
\end{proof}

We introduce the notation $\distinct^n M$ as syntactic sugar for the $\distinct$ operator applied $n$ times to the term $M$ (in particular: $\distinct^0 M = M$). We use it to state and prove the following strengthened version of the reducibility lemma for comprehensions.
\begin{lem}[extends~\ref{lem:redcomprmain}]\label{lem:hredcomprmain}
  Let $K$, $\pure{L}$, $\pure{N}$ be such that $K[p \mapsto \pure{N}\subst{x}{\pure{L}}] \in \SN$ and $\pure{L} \in \SN$. Then for all $n$, $K[p \mapsto \comprehension{\pure{N}|x \gets \distinct^n \setlit{\pure{L}}}] \in \SN$.
\end{lem}
\begin{proof}
Due to the updated statement of this result, we need a stronger metric on \linebreak
$(K,p,\pure{N},\pure{L},n)$:
  \[
  \begin{array}{l}
  (K_1,p_1,\pure{N_1},\pure{L_1},n_1) \prec (K_2,p_2,\pure{N_2},\pure{L_2},n_2)
  \\
  \iff
  (\maxred{K_1[p_1 \mapsto \pure{N_1}\subst{x}{\pure{L_1}}]} + \maxred{\pure{L_1}}, \sz{K_1}_{p_1},\kwsize(\pure{N_1}),n_1)
  \\
  \qquad\quad \lessdot
  (\maxred{K_2[p_2 \mapsto \pure{N_2}\subst{x}{\pure{L_2}}]} + \maxred{\pure{L_2}}, \sz{K_2}_{p_2},\kwsize(\pure{N_2}),n_2)
  \end{array}
  \]
The cases considered in the proof of~\ref{lem:redcomprmain} can be mapped to this extended result in a straightforward manner (however, a reduction to $K[p \mapsto \pure{N}\subst{x}{\pure{L}}]$ is possible only if $n=0$). We also need to consider the following two additional contracta:
\begin{itemize}
  \item $K_0[p \mapsto \comprehension{\distinct\pure{N} \mid x \gets \distinct^{n+1}\setlit{\pure{L}}}]$, where $K = K_0 \compop{p} (\distinct \Box)$: we prove
  \begin{mathpar}
    \maxred{K_0[p \mapsto (\distinct \pure{N})\subst{x}{\pure{L}}]} = \maxred{K[p \mapsto \pure{N}\subst{x}{\pure{L}}]} \quad\text{and}\quad \sz{K_0}_p < \sz{K}_p
  \end{mathpar}
  then the term is s.n.\ by IH\@.
  \item $K[p \mapsto \comprehension{\pure{N} \mid x \gets \distinct^{n-1}\setlit{\pure{L}}}]$, where $n > 0$: since $n-1 < n$ and all of the other values involved in the metric are invariant, we can immediately apply the IH to obtain the thesis. % TODO: Unmatched
  \qedhere
\end{itemize}
\end{proof}

\begin{lem}[extends~\ref{lem:redwheremain}]\label{lem:hredwheremain}
  Let $Q$, $\pure{B}$, $O$ such that $Q[p \mapsto O] \in \SN$, $\pure{B} \in \SN$, $\BV(Q) \cap FV(\pure{B}) = \emptyset$, and $\supp(Q) \cap \supp(O) = \emptyset$. Then $Q[p \mapsto \plwhere~\pure{B}~\kwdo~O] \in \SN$.
\end{lem}
\begin{proof}
We need to consider the following additional contracta of $Q[p \mapsto \plwhere~\pure{B}~\kwdo~P]$:
\begin{itemize}
  \item $Q_0[p \mapsto \plwhere~\pure{B}~\kwdo~\distinct O]$, where $Q = Q_0 \compop{p} (\distinct \Box)$: we show that $\maxred{Q_0[p \mapsto \distinct O]} = \maxred{Q[p \mapsto O]}$ and $\len{Q_0}_p < \len{Q}_p$; then we can apply the IH to prove the term is s.n.
  \item $Q_0[p \mapsto \plwhere~\pure{B}~\kwdo~\promote O]$, where $Q = Q_0 \compop{p} (\promote \Box)$: this is similar to the case above.
  \qedhere
\end{itemize}
\end{proof}

\noindent
Having proved that all the typing rules preserve reducibility, we obtain that all well-typed terms of \NRCldi are strongly normalizing and, as a corollary, the same property holds for \NRClsb.
\begin{thm}\label{thm:nrcldiSN}
If $\Gamma \vdash M : T$ in \NRCldi, then $M \in \SN$ in \NRCldi.
\end{thm}
\begin{cor}\label{cor:nrclsbSN}
If $\Gamma \vdash M : T$ in \NRClsb, then $M \in \SN$ in \NRClsb.
\end{cor}

\section{Related Work}%
\label{sec:related}

This paper builds on a long line of research on normalization of
comprehension queries, a model of query languages popularized over 25
years ago by Buneman et al.~\cite{BNTW95} and inspired by Trinder and Wadler's
work on comprehensions~\cite{trinder-wadler:comprehensions,Wadler92}.  Wong~\cite{wong96jcss}
proved conservativity via a strongly normalizing rewrite system, which
was used in
\linebreak
Kleisli~\cite{wong:comprehensions}, a functional query
system, in which flat query expressions were normalized to SQL\@.
Libkin and Wong~\cite{LW94,LW97} investigated conservativity in the
presence of aggregates, internal generic functions, and bag
operations, and demonstrated that bag operations can be expressed
using nested comprehensions.  However, their normalization results
studied bag queries by translating to relational queries with
aggregation, and did not consider higher-order queries, so they do not
imply the normalization results for $NRC_\lambda(Set,Bag)$
given here.

Cooper~\cite{Cooper09} first investigated query normalization (and
hence conservativity) in the presence of higher-order functions.  He
gave a rewrite system showing how to normalize homogeneous (that is,
pure set or pure bag) queries to eliminate intermediate occurrences of
nesting or of function types.  However, although Cooper
claimed a proof (based on $\top\top$-lifting~\cite{LindleyS05}) and
provided proof details in his PhD thesis~\cite{Cooper09a}, there
unfortunately turned out to be a nontrivial lacuna in that proof, and
this paper therefore (in our opinion) contains the first
\emph{complete} proof of normalization for higher-order queries, even
for the homogeneous case.

Admittedly, our approach is sometimes difficult to work with: the difficulty lies with the notion of (variable capturing) context, along with rewrite rules involving substitutions and renaming of bound variables, as we noted in Section~\ref{sec:red}; for this reason, it would be interesting to consider alternatives. The complexity of computing with contexts has been the object of research in higher-order rewriting and higher-order abstract syntax techniques (\cite{Oostrom93,pfenning88}). Another approach that could be more easily adapted to our scenario is to extend the language to allow hole variables to be decorated with explicit substitutions~(\cite{lambdasigma}). In Section~\ref{sec:red} we have shown that if an unapplied context reduces in a certain way, the same reduction does not have to be allowed when the context is applied to an instantiation. The simplest example we have shown of this phenomenon is:
\[
\begin{array}{c}
(\lambda z.\ibox{p})~N \red \ibox{p}
\\
\text{but}
\\
((\lambda z.\ibox{p})~N)[p \mapsto z] \not\red \ibox{p}[p \mapsto z]
\end{array}
\]
The reason for this discrepancy lies in the fact that while beta reduction yields a substitution replacing $z$ with $N$, once this substitution meets the hole $\ibox{p}$, it is completely lost. If we replaced the meta-operation of substitution with new syntax $L\tuple{x := M}$ denoting a (suspended) explicit substitution that will eventually replace with $M$ all the free occurrences of $x$ within $L$, we could write:
\[
\begin{array}{c}
(\lambda z.\ibox{p})~N \red \ibox{p}\tuple{z := N}
\\
\text{and}
\\
((\lambda z.\ibox{p})~N)[p \mapsto z] \red \ibox{p}\tuple{z := N}[p \mapsto z] = z\tuple{z := N}
\end{array}
\]
where the final term correctly reduces to $N$.
Holes with explicit substitutions have been studied in the context of dependently-typed lambda calculi, where they are more often known as \emph{metavariables}, with applications to proof assistants (\cite{munoz2001}).
We could study strong normalization in such an extended calculus,
however explicit substitutions are known to require a careful treatment of reduction for them to simultaneously preserve
confluence and strong normalization (see~\cite{mellies05} for a counterexample); more recent explicit substitution
calculi (e.g.~\cite{david01,kesner05}) often employ ideas from linear logic to ensure strong normalization is preserved.

Another approach, introduced by Bognar and De Vrijer, employs a \emph{context calculus} (\cite{bognar2001}), i.e.\ an extension of the lambda calculus with additional operators to express context-building and instantiation, along with interfaces describing the evolution of contexts under reduction (``communication''). Under this approach, the context $(\lambda z.\ibox{p})~N$ would be expressed as
\[
\delta \ibox{p}.(\lambda z.\ibox{p}\tuple{z})~N
\]
where the operator $\delta \ibox{p}.-$ (unrelated to the deduplication operator of Section~\ref{sec:heterogeneous}) builds a context by abstracting over a hole variable $\ibox{p}$, and the syntax $\ibox{p}\tuple{z}$ expresses the fact that once $\ibox{p}$ is instantiated with a term, this term will communicate with the context by means of the (captured) variable $z$. The term $z$ to be plugged into the context would be represented as
\[
\Lambda z.z
\]
where the abstraction $\Lambda z.-$ is provided to express the fact that this term can communicate with the context over the variable $z$. To apply this term to the context, we use the syntax $-\lceil - \rceil$:
\[
(\delta \ibox{p}.(\lambda z.\ibox{p}\tuple{z})~N)\lceil \Lambda z.z \rceil
\]
Here as well, we are allowed to beta reduce the context both in the unapplied and in the applied form:
\[
\begin{array}{c}
\delta \ibox{p}.(\lambda z.\ibox{p}\tuple{z})~N \red \delta \ibox{p}.\ibox{p}\tuple{N}
\\
\text{and}
\\
(\delta \ibox{p}.(\lambda z.\ibox{p}\tuple{z})~N)\lceil \Lambda z.z \rceil \red (\delta \ibox{p}.\ibox{p}\tuple{N})\lceil \Lambda z.z \rceil
\end{array}
\]
where the final term can be further reduced to $(\Lambda z.z)\tuple{N}$, and finally to $N$, as expected. Like explicit substitutions, the context calculus allows contexts to be reduced independently of an applied instantiation, potentially simplifying technical results such as those of Lemma~\ref{lem:redtoredctxapp} and~\ref{lem:redrename_aux4}. Both techniques require fairly important extension to the language, type system and rewrite system, and will be considered in future work.

Since the fundamental work of Wong and others on the Kleisli system,
language-integrated query has gradually made its way into other
systems, most notably Microsoft's\ .NET framework languages C\# and % chktex 26
F\#~\cite{meijer:sigmod}, and the Web programming language Links~\cite{CLWY06}.  Cheney et al.~\cite{cheney13icfp}
formally investigated the F\# approach to language-integrated query
and showed that normalization results due to Wong and Cooper could be
adapted to improve it further; however, their work considered only
homogeneous collections. In subsequent work, Cheney et al.~\cite{cheney14sigmod} showed how
use normalization to perform \emph{query shredding} for multiset queries, in which a
query returning a type with $n$ nested collections can be implemented
by combining the results of $n$ flat queries; this has been
implemented in Links~\cite{CLWY06}.

Higher-order relational queries have also been studied by Benedikt et
al.~\cite{benedikt15ic}, where the focus was mostly on complexity of
the evaluation and containment problems.  Their calculus focuses on
higher-order expressions composing operations over \emph{flat}
relational algebra operators only, where the base types are records
listing the fields of the relations.  Thus, modulo notational
differences, their calculus is a sublanguage of \NRC.  In their % chktex 13
setting, normalization up to $\beta$-reduction follows as a special
case of normalization for typed lambda-calculus; in our setting
the same approach would not work because collection and record types
can be combined arbitrarily in \NRC and normalization involves rules
that nontrivially rearrange comprehensions and other collection operations.

Several recent efforts to formalize and reason about the semantics of
SQL are complementary to our work.  Guagliardo and
Libkin~\cite{guagliardo17} presented a semantics for SQL's actual
behaviour in the presence of set and multiset operators (including bag
intersection and difference) as well as
incomplete information (nulls), and related the expressiveness of this
fragment of SQL with that of an algebra over bags with nulls.  Chu et
al.~\cite{Chu17} presented a formalized semantics for reasoning about
SQL (including set and bag semantics as well as aggregation/grouping,
but excluding nulls) using nested relational queries in Coq, while
Benzaken and Contejean~\cite{benzaken19cpp} presented a semantics
including all of these SQL features (set, multiset,
aggregation/grouping, nulls), and formalized the semantics in Coq.
Kiselyov et al.~\cite{kiselyov17aplas} has proposed
language-integrated query techniques that handle sorting operations
(SQL's $\kworder~\kwby$).

However, the above work on semantics has not considered query
normalization, and to the best of our knowledge normalization results
for query languages with more than one collection type were previously
unknown even in the first-order case.  We are interested in extending
our results for mixed set and bag semantics to handle nulls,
grouping/aggregation, and sorting, thus extending higher-order
language integrated query to cover all of the most widely-used SQL
features.  Normalization of higher-order
queries in the presence of all of these features simultaneously
remains an open problem, which we plan to consider next.  In addition,
fully formalizing such normalization proofs also appears to be a
nontrivial challenge.

\section{Conclusions}%
\label{sec:concl}

Integrating database queries into programming languages has many
benefits, such as type safety and avoidance of common SQL injection
attacks, but also imposes limitations that prevent programmers from
constructing queries dynamically as they could by concatenating SQL
strings unsafely.  Previous work has demonstrated that many useful
dynamic queries can be constructed safely using \emph{higher-order
  functions} inside language-integrated queries; provided such
functions are not recursive, it was believed, query expressions can be
normalized.  Moreover, while it is common in practice for
language-integrated query systems  to provide
support for SQL features such as mixed set and bag operators, it is
not well understood in theory how to normalize these queries in the
presence of higher-order functions.
Previous work on higher-order query normalization has considered only
homogeneous (that is, pure set or pure bag) queries, and in the
process of attempting to generalize this work to a heterogeneous
setting, we discovered a nontrivial gap in the previous proof of
strong normalization.  We therefore prove strong normalization for
both homogeneous and heterogeneous queries for the first time.

As next steps, we intend to extend the Links implementation of
language-integrated query with heterogeneous queries and
normalization, and to investigate (higher-order) query normalization
and conservativity for the remaining common SQL features, such as
nulls, grouping/aggregation, and ordering.

\section*{Acknowledgments}
This work was supported by ERC Consolidator Grant Skye (grant number
ERC-682315) and by an ISCF Metrology
Fellowship grant provided by the UK government’s Department for
Business, Energy and Industrial Strategy (BEIS).

This research has been supported by the National Cyber Security Centre (NCSC) project: Mechanising the metatheory of SQL with nulls.

We are grateful to
Philip Wadler, Sam Lindley, and the anonymous reviewers for their
comments and suggestions.

\bibliography{paper}
\bibliographystyle{alphaurl}

\newpage
\appendix

\section{Proofs}
This appendix expands on some results whose proofs were omitted or only sketched in the paper.

\begin{lemmaproof}{lem:Kredempty}{
  If $K \in \SN$ is a continuation, then for all indices $p$ we have $K[p \mapsto \emptyoset] \in \SN$.
}
    We proceed by well-founded induction, using the metric:
    \[
    (K_1,p_1) \prec (K_2,p_2) \iff (\maxred{K_1},\sz{K_1}_{p_1}) \lessdot (\maxred{K_2},\sz{K_2}_{p_2})
    \]
    %We only consider the case for $p \in \supp(K)$ (otherwise, the proof is trivial), and show that all of the contracta of $K[p \mapsto \emptyoset]$ must be s.n.:
    \begin{itemize}
    \item $K'[p \mapsto \emptyoset]^\sigma$, where $K \ared{\sigma} K'$: by Corollary~\ref{cor:aredforall}, we need to show $K'[q \mapsto \emptyoset] \in \SN$ whenever $\sigma(q) = p$; this follows from the IH, with $\maxred{K'} < \maxred{K}$ by Lemma~\ref{lem:redmaxred}.
    \item $K_0[p \mapsto \emptyoset]$, where $K = K_0 \compop{p} F$ for some frame $F$: by Lemma~\ref{lem:maxredsubK} we have $\maxred{K_0} \leq \maxred{K}$; furthermore, by Lemma~\ref{lem:szdecrease} we show that $\sz{K_0}_p < \sz{K}_p$; then the thesis follows immediately from the IH\@.
    \qedhere
    \end{itemize}
\end{lemmaproof}

\begin{lemmaproof}{lem:Qredunion}{
For all Q-continuations $Q,O_1,O_2$ with pairwise disjoint supports, if $Q[p \mapsto O_1] \in \SN$ and $Q[p \mapsto O_2] \in \SN$, then $Q[p \mapsto O_1 \ocup O_2] \in \SN$.
}
We assume $p \in \supp(Q)$ (otherwise, $Q[p \mapsto O_1] = Q[p \mapsto O_2] = Q[p \mapsto O_1 \ocup O_2]$, and the thesis holds trivially).
Then, by Lemma~\ref{lem:ctxappsn}, $Q[p \mapsto O_1] \in \SN$ and $Q[p \mapsto O_2] \in \SN$ imply $Q \in \SN$, $O_1 \in \SN$, and $O_2 \in \SN$:
thus we can proceed by well-founded induction on $(Q,p,O_1,O_2)$ using the following metric:
\[
\begin{array}{l}
(Q^1, p^1, O^1_1, O^1_2) \prec (Q^2, p^2, O^2_1,O^2_2)
\\
\iff
(\maxred{Q^1},\sz{Q^1}_{p^1},\maxred{O^1_1} + \maxred{O^1_2}) \lessdot
(\maxred{Q^2},\sz{Q^2}_{p^2},\maxred{O^2_1} + \maxred{O^2_2})
\end{array}
\]
to prove that if $Q[p \mapsto O_1] \in \SN$ and $Q[p \mapsto O_2] \in \SN$, then $Q[p \mapsto O_1 \ocup O_2] \in \SN$.
Equivalently, we will consider all possible contracta and show that each of them must be a strongly normalizing term; we will apply the induction hypothesis to new auxiliary continuations obtained by placing pieces of $Q$ into $O_1$ and $O_2$: the hypothesis on the supports of the continuations being disjoint is used to make sure that the new continuations do not contain duplicate holes and are thus well-formed. By cases on the possible contracta:

\begin{itemize}
\item $Q_1 [q \mapsto Q_2\subst{x}{\pure{L}}][p \mapsto (O_1\subst{x}{\pure{L}}) \ocup (O_2\subst{x}{\pure{L}})]$ (where \[Q =
(Q_1 \compop{q} \comprehension{\Box \mid x \gets \setlit{\pure{L}}})[q \mapsto Q_2],\] $q \in \supp(Q_1)$, $p \in \supp(Q_2)$): let $Q' = Q_1[q \mapsto Q_2\subst{x}{\pure{L}}]$, and note that $Q \red Q'$, hence $\maxred{Q'} < \maxred{Q}$; note $Q[p \mapsto O_1] \red Q'[p \mapsto O_1\subst{x}{\pure{L}}]$, hence since the former term is s.n., so must be the latter, and hence also $O_1\subst{x}{\pure{L}} \in \SN$; similarly, $O_2\subst{x}{\pure{L}}$; then we can apply the IH with
$(Q',p,O_1\subst{x}{\pure{L}},O_2\subst{x}{\pure{L}})$ to obtain the thesis.
\item $Q'[p \mapsto O_1 \ocup O_2]^\sigma$ (where $Q \ared{\sigma} Q'$): by Corollary~\ref{cor:aredforall}, we need to prove that, for all $q$ s.t.\ $\sigma(q) = p$, $Q'[q \mapsto O_1 \ocup O_2] \in \SN$; since $Q[p \mapsto O_1] \in \SN$, we also have $Q'[p \mapsto O_1]^\sigma \in \SN$, which implies $Q'[q \mapsto O_1] \in \SN$ by Corollary~\ref{cor:aredforall}; for the same reason, $Q'[q \mapsto O_2] \in \SN$; by Lemma~\ref{lem:redmaxred}, $\maxred{Q'} < \maxred{Q}$, thus the thesis follows by IH\@.
\item $Q_1 [p \mapsto (\comprehension{Q_2|x \gets O_1}) \ocup (\comprehension{Q_2|x \gets O_2})]$ (where $Q = Q_1 \compop{p} \comprehension{Q_2|x \gets \Box}$):
\linebreak
by Lemma~\ref{lem:maxredsubK}, $\maxred{Q_1} \leq \maxred{Q}$; we also know $\sz{Q_1}_p < \sz{Q}_p$; take $O_1' := \comprehension{Q_2|x \gets O_1}$ and note that, since
$Q[p \mapsto O_1] = Q_0[p \mapsto O_1']$, we have $O_1'$ is a subterm of a strongly normalizing term, thus $O_1' \in \SN$; similarly, we define $O_2' := \comprehension{Q_2|x \gets O_2}$ and show it is s.n.\ in a similar way; then $(Q_1,p,O_1', O_2')$ reduce the metric, and we can prove the thesis by IH\@.
\item $Q_1 [p \mapsto (\comprehension{O_1|x \gets Q_2}) \ocup (\comprehension{O_2|x \gets Q_2})]$ (where $Q = Q_1 \compop{p} \comprehension{\Box \mid x \gets Q_2}$):
\linebreak
by Lemma~\ref{lem:maxredsubK}, $\maxred{Q_1} \leq \maxred{Q}$; by Lemma~\ref{lem:szdecrease} we also know $\sz{Q_1}_p < \sz{Q}_p$; take $O_1' := \comprehension{O_1|x \gets Q_2}$ and note that, since
$Q[p \mapsto O_1] = Q_1[p \mapsto O_1']$, we have $O_1'$ is a subterm of a strongly normalizing term, thus $O_1' \in \SN$; similarly, we define $O_2' := \comprehension{O_2|x \gets Q_2}$ and show it is s.n.\ in a similar way; then $(Q_1, p, O_1', O_2')$ reduce the metric, and we can prove the thesis by IH\@.
\item $Q_0 [p \mapsto (\plwhere~\pure{B}~\kwdo~O_1) \ocup (\plwhere~\pure{B}~\kwdo~O_2)]$ (where $Q = Q_0 \compop{p} (\plwhere~\pure{B}~\kwdo~\Box)$):
by Lemma~\ref{lem:maxredsubK}, $\maxred{Q_0} \leq \maxred{Q}$; by Lemma~\ref{lem:szdecrease} we also know $\sz{Q_0}_p < \sz{Q}_p$; take $O_1' := \plwhere~B~\kwdo~O_1$ and note that, since $Q[p \mapsto O_1] = Q_0[p \mapsto O_1']$, we have $O_1'$ is a subterm of a strongly normalizing term, thus $O_1' \in \SN$; similarly, we define $O_2' := \plwhere~\pure{B}~\kwdo~O_2$ and prove it is strongly normalizing in the same way; then $(Q_0, p, O_1',O_2')$ reduce the metric, and we can prove the thesis by IH\@.
% BAGS
%\item $Q_0 [p \mapsto (\distinct M) \ocup (\distinct N)]$ (where $Q = Q_0 \compop{p} (\distinct \Box)$):
%by Lemma~\ref{lem:maxredsubK}, $\maxred{Q_0} \leq \maxred{Q}$; we also know $\sz{Q_0}_p < \sz{Q}_p$; take $M' := \distinct M$ and note that, since $Q[p \mapsto M] = Q_0[p \mapsto M']$, we have $M'$ is a subterm of a strongly normalizing term, thus $M' \in \SN$; similarly, we define $N' := \distinct N$ and prove it is strongly normalizing in the same way; then $(Q_0, p, M',N')$ reduce the metric, and we can prove the thesis by IH.
\item Contractions within $O_1$ or $O_2$ reduce $\maxred{O_1} + \maxred{O_2}$, thus the thesis follows by IH\@.
\qedhere
\end{itemize}
\end{lemmaproof}

\noindent
Reducibility for conditionals is proved similarly to comprehensions. However, to consider all the conversions commuting with $\plwhere$, we need to use the more general auxiliary continuations.
\begin{lemmaproof}{lem:redwherechange}{
Suppose $Q[p \mapsto \plwhere~B~\kwdo~M] \in \SN$. Then for all $B' \in \SN$ such that $\BV(Q)$ and $\FV(B')$ are disjoint, $Q[p \mapsto \plwhere~B \land B'~\kwdo~M] \in \SN$.
}
  We proceed by well-founded induction on $(Q,B,B',M,p)$ using the following metric:
  \[
    \begin{array}{l}
    (Q_1,B_1,B'_1,M_1,p_1) \prec (Q_2,B_2,B'_2,M_2,p_2)
    \iff
    \\
    \qquad
    (\maxred{Q_1[p_1 \mapsto \plwhere~B_1~\kwdo~M_1]},\maxred{B'_1},\kwsize(M_1))
    \\
    \qquad
    \lessdot
    (\maxred{Q_2[p_2 \mapsto \plwhere~B_2~\kwdo~M_2]}, \maxred{B'_2}, \kwsize(M_2))
    \end{array}
  \]
  We will consider all possible contracta of $Q[p \mapsto \plwhere~B \land B'~\kwdo~M]$ and show that each of them must be a strongly normalizing term. By cases:
  \begin{itemize}
  \item $Q_1[q \mapsto Q_2 \subst{x}{\pure{L}}][p \mapsto (\plwhere~B \land B'~\kwdo~M)\subst{x}{\pure{L}}]$, where \[Q =
  (Q_1 \compop{q} \comprehension{\Box \mid x \gets \setlit{\pure{L}}})[q \mapsto Q_2],\] $q \in \supp(Q_1)$, and $p \in \supp(Q_2)$; by the freshness condition we know $x \notin \FV(B')$, thus $(\plwhere~B \land B'~\kwdo~M)\subst{x}{\pure{L}} = \plwhere~B\subst{x}{\pure{L}} \land B'~\kwdo~(O\subst{x}{\pure{L}})$; to apply the IH, we need to show
  $\maxred{Q_1[q \mapsto Q_2 \subst{x}{\pure{L}}][p \mapsto \plwhere~B\subst{x}{\pure{L}}~\kwdo~M]} < \maxred{Q[p \mapsto \plwhere~B~\kwdo~M]}$: since the former term is a contractum of the latter, this is implied by Lemma~\ref{lem:redmaxred}.
  \item $Q'[p \mapsto \plwhere~B \land B'~\kwdo~M]^\sigma$, where $Q \ared{\sigma} Q'$. By Corollary~\ref{cor:aredforall}, it suffices to prove $Q'[p \mapsto \plwhere~B \land B'~\kwdo~M]$ for all $q$ s.t.\ $\sigma(p) = q$; we prove $\maxred{Q'[q \mapsto \plwhere~B~\kwdo~M]} \leq \maxred{Q'[p \mapsto \plwhere~B~\kwdo~M]^\sigma}$ (by Corollary~\ref{cor:aredforall}), and $\maxred{Q'[p \mapsto \plwhere~B~\kwdo~M]^\sigma} < \maxred{Q[p \mapsto \plwhere~B~\kwdo~M]}$ (by Lemma~\ref{lem:redmaxred}, since the former term is a contractum of the latter); then the thesis follows by IH\@.
  \item $Q_1[p \mapsto \plwhere~B \land B'~\kwdo~\comprehension{Q_2|x \gets M}]$, where $Q = Q_1 \compop{p} \comprehension{Q_2|x \gets \Box}$; to apply the IH, we need to show $\maxred{Q_1[p \mapsto \plwhere~B~\kwdo~\comprehension{Q_2|x \gets M}]} < \maxred{Q[p \mapsto \plwhere~B~\kwdo~M]}$: since the former term is a contractum of the latter, this is implied by Lemma~\ref{lem:redmaxred}.
  \item $Q_0[p \mapsto \plwhere~B_0 \land B \land B'~\kwdo~O]$, where $Q = Q_0 \compop{p} (\plwhere~B_0~\kwdo~\Box)$; to apply the IH, we need to show $\maxred{Q_1[p \mapsto \plwhere~B_0 \land B~\kwdo~M]} < \maxred{Q[p \mapsto \plwhere~B~\kwdo~M]}$: since the former term is a contractum of the latter, this is implied by Lemma~\ref{lem:redmaxred}.
  \item $Q[p \mapsto \emptyoset]$, where $O = \emptyoset$: this term is also a contractum of $Q[p \mapsto \plwhere~B~\kwdo~\emptyoset]$, thus it is strongly normalizing.
  \item $Q[p \mapsto (\plwhere~B \land B'~\kwdo~M_1) \ocup (\plwhere~B \land B'~\kwdo~M_2)]$, where $M = M_1 \ocup M_2$; we note that,
  % by Lemma~\ref{lem:Kredunionr},
  for $i=1,2$, we have $\maxred{Q[p \mapsto \plwhere~B~\kwdo~M_i]} \leq \maxred{Q[p \mapsto (\plwhere~B~\kwdo~M_1) \cup (\plwhere~B~\kwdo~M_2)]} < \maxred{Q[p \mapsto \plwhere~B~\kwdo~M]}$, where the first inequality is by Lemma~\ref{lem:Kredunionr}, and the second by Lemma~\ref{lem:redmaxred}; we also note $\kwsize(M_i) < \kwsize(M)$; then we can apply the IH to prove $Q[p \mapsto \plwhere~B \land B'~\kwdo~M_i] \in \SN$, which implies the thesis by Lemma~\ref{lem:Kredunion}.
  \item $(Q \compop{p} \comprehension{\Box \mid x \gets M_2})[p \mapsto \plwhere~B \land B'~\kwdo~M_1]$, where $M = \comprehension{M_1|x \gets M_2}$;
  to apply the IH, we need to show $\maxred{(Q \compop{p} \comprehension{\Box \mid x \gets M_2})[p \mapsto \plwhere~B~\kwdo~M_1]} < \maxred{Q[p \mapsto \plwhere~B~\kwdo~M]}$: since the former is a contractum of the latter, this is implied by Lemma~\ref{lem:redmaxred}.
  \item $Q[p \mapsto \plwhere~B \land B' \land B_0~\kwdo~M_0]$, where $M = \plwhere~B_0~\kwdo~M_0$;
  to apply the IH, we need to show $\maxred{Q[p \mapsto \plwhere~B \land B_0~\kwdo~M_0]} < \maxred{Q[p \mapsto \plwhere~B~\kwdo~M]}$: since the former is a contractum of the latter, this is implied by Lemma~\ref{lem:redmaxred}.
  \item Reductions within $B$ or $M$ make the induction metric smaller, thus follow immediately from the IH\@.
  \qedhere
  \end{itemize}
\end{lemmaproof}
\end{document}